\documentclass[11pt]{article}
\usepackage[T1]{fontenc}
\usepackage[utf8]{inputenc}

\usepackage[english]{babel}
\usepackage{geometry}
\usepackage{nicefrac}
\usepackage{amsmath, amssymb, amscd, amsthm, amsfonts}
\usepackage[hidelinks]{hyperref}
\usepackage{hyperref,natbib}
\usepackage{xcolor}
\usepackage{tcolorbox}
\usepackage{algorithm}
\usepackage{algorithmic}
\usepackage{float}
\usepackage{multirow}
\usepackage{caption}
\usepackage{booktabs, tabularx}
\usepackage{color}
\usepackage{tikz}
\usetikzlibrary{math}
\usepackage{pgfplots}
\usepackage{booktabs}
\usepackage{graphicx}
\usepackage{placeins}
\usepackage{booktabs,multirow,tabularx,array}
\hypersetup{
    hidelinks
}
\pgfplotsset{compat=1.17}

\renewcommand{\arraystretch}{1.1}
\newcolumntype{C}{>{\centering\arraybackslash}X}

\newtcolorbox{boxA}{
    fontupper = \bfseries,
    boxrule = 1.5pt,
    colframe = black,
    colback = white 
}
\hypersetup{
  pdftitle={},
  pdfauthor={},
  pdfsubject={},
  pdfkeywords={}
}
\makeatletter

\def\moverlay{\mathpalette\mov@rlay}
\def\mov@rlay#1#2{\leavevmode\vtop{%
   \baselineskip\z@skip \lineskiplimit-\maxdimen
   \ialign{\hfil$\m@th#1##$\hfil\cr#2\crcr}}}
\newcommand{\charfusion}[3][\mathord]{
    #1{\ifx#1\mathop\vphantom{#2}\fi
        \mathpalette\mov@rlay{#2\cr#3}
      }
    \ifx#1\mathop\expandafter\displaylimits\fi}
\makeatother

\theoremstyle{plain}
\newtheorem{thm}{\protect\theoremname}[section]
\theoremstyle{definition}

\theoremstyle{plain}
\newtheorem{lemma}[thm]{\protect\lemmaname}
\newtheorem{corollary}[thm]{\protect\corollaryname}

\newtheorem{example}[thm]{Example}
\newtheorem{proposition}[thm]{\protect\propositionname}
\newtheorem{remark}[thm]{\protect\remarkname}
\newtheorem{claim}[thm]{\protect\claimname}

\newcommand{\OPT}{{\operatorname{OPT}}}

\newcommand{\N}{{\mathcal{N}}}
\newcommand{\cov}{{\mathrm{cov}}}
\newcommand{\het}{{\mathrm{het}}}
\newcommand{\hhom}{{\mathrm{hom}}}
\newcommand{\cent}{{\mathrm{cent}}}

\makeatother

\usepackage{babel}
\usepackage{graphicx}
\providecommand{\definitionname}{Definition}
\providecommand{\lemmaname}{Lemma}
\providecommand{\theoremname}{Theorem}
\providecommand{\corollaryname}{Corollary}
\providecommand{\questionname}{Question}
\providecommand{\notename}{Note}
\providecommand{\propositionname}{Proposition}
\providecommand{\remarkname}{Remark}
\providecommand{\claimname}{Claim}

\begin{document}
\title{Stochastic Multi-Robot Monitoring on Graphs under Markovian Mobility}
\author{%
Walid Ben-Ameur \qquad Tijani Chahed \qquad Shamisa Nematollahi\thanks{Emails:
\texttt{\{walid.benameur,tijani.chahed,shamisa.nematollahi\}@telecom-sudparis.eu}}\\[0.5em]
\small Télécom SudParis, Institut Polytechnique de Paris, SAMOVAR,
Palaiseau, France%
}

\date{}


\maketitle

\begin{abstract}
We study a stochastic multi-robot monitoring problem on a connected graph $G=(V,E)$, where each robot moves according to a Markov chain on $G$ and monitors the closed neighborhood of its current vertex. The performance of $r$ robots is evaluated in steady state via two objectives: average-case coverage (the expected number of covered vertices) and worst-case coverage (the minimum coverage probability over all vertices). We consider three models: independent homogeneous strategies, where all robots share the same stationary distribution; independent heterogeneous strategies, where robots use different stationary distributions; and centralized strategies, allowing arbitrary correlations between robot locations.

For the heterogeneous model, we prove that maximizing average coverage is NP-hard even for two robots, and that replicating an easy-to-compute optimal homogeneous strategy yields a \(\left(1-\left(1-\frac{1}{r}\right)^r\right)\)-approximation for both objective functions in the heterogeneous setting; moreover, 
no polynomial-time algorithm can achieve a ratio better than \(1-\nicefrac{1}{e}\) unless \(\text{P}=\text{NP}\).

Centralized strategies can exploit correlations to reduce redundancy. We develop a hierarchy of approximation factors:
for any positive integer \(r'\le r\), writing \(r=hr'+b\) with \(0\le b<r'\), block coordination yields a \(1-\left(1-\frac{r'}{r}\right)^h\left(1-\frac{b}{r}\right)\)
approximation for both objectives. We also establish NP-hardness and a tight \(1-\nicefrac{1}{e}\) inapproximability bound. Moreover, we prove diminishing-returns properties with respect to the number of robots: a non-increasing-ratio property holds for the average-case objective in all settings, but not for the heterogeneous worst-case objective. These results provide a unified complexity and approximation landscape for stochastic graph monitoring and quantify the benefit of heterogeneity and coordination in steady-state multi-robot monitoring.
\end{abstract}



\section{Introduction}
\label{sec:intro}
Large facilities such as industrial plants, tunnels, warehouses, and transportation networks must often be monitored continuously for rare but critical events, such as a fire starting in a remote corridor, an intruder entering a restricted area, or a sensor anomaly indicating a local failure. Since a robot stationed at a fixed location can observe only a limited surrounding area, some regions may remain poorly monitored. To use a limited number of robots more efficiently, it is natural to let them move through the environment over time, so that they collectively provide broader and more reliable coverage. This naturally leads to a stochastic multi-robot monitoring problem on a graph, where vertices represent locations of interest, edges represent feasible robot movements, and a robot observes the neighborhood of its current position. Motivated by this setting, we study the following problem. Let \(G=(V,E)\) be a connected undirected graph with \(n\) vertices, representing the environment to be monitored by \(r\) robots.
We adopt a \emph{wait-action} mobility model: during each time step, a may either move to an adjacent vertex or remain at its current vertex, represented by a self-loop at every vertex.
Each edge traversal requires one time slot.
A robot  located at a vertex $u$ can observe all vertices in the closed neighborhood $\N[u]$. 
We assume that each robot evolves according to a Markov chain whose transition matrix is chosen subject to the mobility constraints, and we denote its stationary distribution by \(\mu_i\).
Thus, the coverage probability of a vertex $v$ by  the $i$-th robot is
$\cov(v,\mu_i) = \sum_{u \in \N[v]} \mu_i(u)$.
Assuming that the robots move independently, the coverage probability of a vertex $v$ is
$\cov(v, \mu_1,...,\mu_r) = 1 - \prod_{i=1}^r \bigl(1 - \cov(v,\mu_i)\bigr)$.
We study two objectives: \emph{average coverage}, defined as the expected number of covered vertices, \(\sum_{v\in V}\cov(v,\mu_1,\ldots,\mu_r)\), \emph{worst-case coverage}, defined as the minimum monitoring probability \(\min_{v\in V}\cov(v,\mu_1,\ldots,\mu_r)\).
When robots are allowed distinct stationary distributions, we call this the \emph{independent heterogeneous} setting. Let $F_{\mathrm{het}}(r,\mu_1,\ldots,\mu_r)$ denote the average coverage. \sloppy{ Our objective is $F^*_{\mathrm{het}}(r)=\max_{\mu_i\in\mathcal{P}} F_{\mathrm{het}}(r,\mu_1,\ldots,\mu_r)$, where $\mathcal{P}$ is the feasible polytope for each $\mu_i$ (defined later).} Similarly, let $W_{\mathrm{het}}(r,\mu_1,\ldots,\mu_r)$ denote worst-case coverage, with optimum $W^*_{\mathrm{het}}(r)$.

When the stationary distributions of the robots are required to be identical, say $\mu$, while their movements remain independent, we refer to this setting as the \emph{independent homogeneous} case. 
The average (resp.\ worst-case) objective function is then denoted by 
$F_{\mathrm{hom}}(r, \mu)$ (resp.\ $W_{\mathrm{hom}}(r, \mu)$), 
and the corresponding optimal values are denoted by 
$F^*_{\mathrm{hom}}(  r)$ and $W^*_{\mathrm{hom}}(  r)$.
Observe that 
$F_{\mathrm{hom}}(r, \mu) = F_{\mathrm{het}}(r, \mu_1 = \mu, \ldots, \mu_r = \mu)$, 
and an analogous relation holds for the worst-case objective function.

The heterogeneous strategy offers greater flexibility than the homogeneous one. 
A natural question is to quantify the benefit of allowing heterogeneity, or equivalently, the loss incurred when restricting to homogeneous strategies. 
Both strategies can be compared to a \emph {centralized} approach that removes the independence restriction and allows an arbitrary joint steady-state distribution over robot configurations. 
This permits correlations among robot locations and may reduce redundancy in monitoring. 
Similarly to the previous notation, the optimal objective values can be denoted by $F^*_{\mathrm{cent}}(  r)$ and $W^*_{\mathrm{cent}}(  r)$.


\subsection{Our Contribution and Technical Overview}
We now review the main results and sketch the techniques used to establish them.
Table~\ref{tab:results-landscape} summarizes the resulting complexity and approximation landscape across the three strategy classes and the two objectives.

\begin{table}[t]
\centering
\renewcommand{\arraystretch}{1.10}
\scalebox{0.9}{%
\begin{tabular}{l|cc|cc}
\hline
\multirow{2}{*}{Model}
&
\multicolumn{2}{c}{Average coverage}
&
\multicolumn{2}{c}{Worst-case coverage}
\\
\cline{2-5}
&
Upper bound
&
Lower bound
&
Upper bound
&
Lower bound
\\
Homogeneous
&
Convex optimization
&
--
&
Linear programming
&
--
\\
Heterogeneous
&
\(\alpha_r-\varepsilon\)
&
\((1-1/e)\)
&
\(\alpha_r\)
&
\((1-1/e)\)
\\
Centralized
&
\(\alpha_{r,r'}-\varepsilon\)
&
\((1-1/e)\)
&
\(\alpha_{r,r'}\)
&
\((1-1/e)\)
\\

\end{tabular}
}
\caption{Summary of the main algorithmic and hardness results. Here
\(\alpha_r=1-(1-1/r)^r\), and for \(r=hr'+b\), \(0\le b<r'\),
\(\alpha_{r,r'}=1-\left(1-\frac{r'}{r}\right)^h\left(1-\frac{b}{r}\right)\).}
\label{tab:results-landscape}
\end{table}

For the independent homogeneous model, Section~\ref{sec:hom} shows that the worst-case objective \(W^*_{\hom}(r)\) can be computed in polynomial time by linear programming and the average-case objective \(F^*_{\hom}(r)\) by convex optimization. 
We also prove that both \(F^*_{\hom}(r)\) and \(W^*_{\hom}(r)\) satisfy a decreasing-ratio property with respect to \(r\), thereby establishing a diminishing-returns phenomenon in the homogeneous setting; see Proposition~\ref{prop:hom_star_concavity}.
We then exploit graph symmetries to reduce the dimension of the optimization problem; see Proposition~\ref{thm:orbit-symmetric}. A connection with fractional domination, when the polytope $\mathcal{P}$ is the standard simplex, is established in Proposition \ref{pro:homfrac}.


For the independent heterogeneous model, Section~\ref{sec:het} reveals a sharp contrast between exact optimization and approximation. 
Theorem~\ref{th:r=2} proves that computing \(F^*_{\het}(r)\) is NP-hard even for \(r=2\), via a reduction from \textsc{Densest \(k\)-Subgraph}. 
On the positive side, Proposition~\ref{prop:exist-poly-fixed-r} shows that \(F^*_{\het}(r)\) is polynomial-time solvable whenever \(r\) is fixed and the feasible polytope \(\mathcal P\) has polynomially many extreme points; the proof uses the multilinearity of the objective and the fact that an optimum is attained at extreme points.
We now go beyond NP-hardness and prove approximation hardness for both objectives. 
Theorems~\ref{thm:inapp-f-het} and~\ref{thm:inapp-het-worst} show that \(F^*_{\het}(r)\) and \(W^*_{\het}(r)\) are APX-hard; no polynomial-time algorithm can achieve an approximation ratio better than \(1-\nicefrac{1}{e}\), unless \(\mathrm{P}=\mathrm{NP}\). 
We then compare heterogeneous and homogeneous monitoring.
Theorem~\ref{thm:het-avg-alpha} and Corollary~\ref{cor:het-worst-alpha} show that replicating an optimal homogeneous strategy yields \(\alpha_r\)-approximation for both heterogeneous objectives, for any polytope \(\mathcal P\), where \(\alpha_r := 1-\left(1-\frac{1}{r}\right)^r\). The proof is based on a pointwise comparison of coverage probabilities.
Finally, Proposition~\ref{prop:starconcavity-het-f} proves a diminishing return property for \(F^*_{\het}(r)\), whereas Example~\ref{exp:concavity-het-worst} shows that the analogous property fails for \(W^*_{\het}(r)\).

For the centralized model, Section~\ref{sec:centralized} allows arbitrary joint steady-state distributions over robot configurations and therefore captures the full benefit of correlation. 
By a direct reduction from \textsc{Dominating Set}, both centralized objectives are NP-hard.
On the algorithmic side, we introduce a hierarchy of approximation factors based on block coordination: 
for any positive integer \(r'\le r\), writing \(r=hr'+b\) with \(0\le b<r'\), we partition the robots into \(h\) groups of size \(r'\) and one remaining group of size \(b\), allow arbitrary coordination within each group of size \(r'\), and repeat the same \(r'\)-robot strategy independently across the full groups, together with the induced strategy for the remaining group. 
This yields an \(\alpha_{r,r'}\)-approximation for both centralized objectives, where \(\alpha_{r,r'}:=1-\left(1-\frac{r'}{r}\right)^h \left(1-\frac{b}{r}\right)\) (Theorem~\ref{thm:cent-block-approx}), and recovers the homogeneous \(\alpha_r\)-approximation as the special case \(r'=1\).
This approximation guarantee is again best possible up to \(1-\nicefrac{1}{e}\); Proposition~\ref{prop:cent-avg-inapprox} and Theorem~\ref{thm:cent-worst-inapprox}.
We also prove that both \(F^*_{\cent}(r)\) and \(W^*_{\cent}(r)\) satisfy a decreasing-ratio property; Proposition~\ref{prop:cent-star-concavity}.

Taken together, the results summarized in Table~\ref{tab:results-landscape} provide a complete picture of the tradeoff between tractability, approximation, and coordination.
From a technical standpoint, the analysis combines concavity arguments for the homogeneous model, multilinear optimization over polytopes for the heterogeneous model, and pointwise coverage comparisons that allow us to transfer homogeneous solutions to the heterogeneous. 

From an implementation standpoint, our results support the use of simple non-centralized homogeneous strategies: they achieve a performance guarantee within a factor of \(1-\nicefrac{1}{e}\) of the fully centralized optimum, while being substantially easier to deploy.

\subsection{Related Work}
Our problem is closely connected to several classical graph covering and domination problems in the deterministic limit. If each robot is assigned deterministically to a single vertex, equivalently if each $\mu_i$ is restricted to be an extreme point of the standard simplex, then the average-case objective reduces to maximizing the size of the dominated set $|\N[S]|$ with $|S|\le r$. This is exactly a maximum-coverage problem on the family of closed neighborhoods, and can also be viewed as a budgeted or partial dominating-set formulation on graphs. Likewise, asking whether the worst-case coverage is equal to~1 is exactly the dominating-set problem, since it asks whether the selected robot positions dominate every vertex. 
Maximum coverage (and maximum domination on closed neighborhoods) admits the classical greedy ratio \(1-\nicefrac{1}{e}\), and this factor is optimal unless \(\text{P}=\text{NP}\)~\cite{feige1998threshold,MIYANO}. Dominating set admits logarithmic-factor approximations and this logarithmic behavior is essentially unavoidable~\cite{feige1998threshold}. Connected dominating set admits approximation guarantees \(H(\Delta)+2\) and \(2H(\Delta)+2\), where $H(\Delta)$ denotes the $\Delta$-th harmonic number and \(\Delta\) is the maximum degree, while the partial and budgeted connected variants admit \(O(\log \Delta)\)- and \(\frac{1}{13}(1-\nicefrac{1}{e})\)-approximations, respectively~\cite{guha1998connected,khuller2020optimal}. Unlike these deterministic formulations, our paper studies the stochastic setting; robots move over time and coverage is quantified through stationary probabilities. Another related line of research considers barrier coverage with mobile sensors~\cite{bhattacharya2009optimal}, which studied the problem of optimally relocating a set of mobile sensors to the boundary of a planar region in order to achieve barrier coverage while minimizing either the maximum or the total movement of the sensors.

The art-gallery problem asks for a minimum number of guards whose visibility regions cover the whole environment. In the mobile variant, the watchman-route problem, one replaces static guards with a moving guard and asks for a route from which every point becomes visible~\cite{urrutia2000artgallery,chin1988optimum}. The point-guard art gallery problem in simple polygons admits an \(O(\log \OPT)\)-approximation under the standard integer-coordinate and general-position assumptions, while guarding polygons with holes is logarithmically inapproximable~\cite{bonnet2017artgallery,eidenbenz2001inapprox}. For the watchman route, simple polygons admit polynomial-time exact algorithms, whereas polygons with holes yield an NP-hard problem with an \(O(\log^2 n)\)-approximation and logarithmic-factor inapproximability~\cite{carlsson1999watchman,mitchell2013watchman}.

Patrolling asks a robot team to repeatedly visit a set of locations so as to keep the time between consecutive visits small. In the graph version, the environment is modeled by a graph and robots move along edges; the goal is to minimize measures such as idleness, refresh time, or latency, that is, the worst or average delay between two visits to the same vertex or region~\cite{chevaleyre2004theoretical,pasqualetti2012cooperative}.
 On the algorithmic side, cooperative patrolling admits polynomial-time algorithms on paths, on trees when the number of robots is fixed, and becomes NP-hard on cyclic graphs, where an \(8\)-approximation is known for minimum refresh time~\cite{pasqualetti2012cooperative}. 
Other graph-patrolling variants include frequency-constrained patrolling, where vertices must be revisited at prescribed rates, and patrolling with distributed or route-based policies~\cite{elmaliach2009multirobot,portugal2014finding}. In persistent monitoring, the objective is to control staleness or weighted latency of observations~\cite{smith2010currency,smith2012persistent,alamdari2014persistent}, and in stochastic surveillance, robot motion is randomized and optimized through Markov-chain criteria such as hitting-time or entropy-based objectives~\cite{grace2005stochastic,srivastava2009stochastic,patel2015robotic,george2019markov}.

\paragraph{Relation to submodular maximization.}
Under simplex constraints, the deterministic average-case special case of our model reduces exactly to Maximum Domination, equivalently, coverage by closed neighborhoods, and hence falls within the classical monotone submodular framework. In that regime, the standard greedy algorithm yields the optimal $(1-\nicefrac{1}{e})$ approximation ratio~\cite{nemhauser1978submodular}.
In the stochastic setting considered here, in the homogeneous model, the relevant structure is concavity in the stationary distribution; in the heterogeneous average-case model, the objective is multilinear and exhibits a continuous diminishing-returns property~\cite{bian2017continuous}; and in the centralized model, the approximation analysis relies on pointwise coverage comparisons via averaged marginals. By contrast, the worst-case objectives are max--min in nature, and the heterogeneous worst-case objective does not even satisfy the non-increasing-ratio property in general.
Thus, while the special case for the deterministic simplex-constrained average-case falls within the classical submodular maximization framework, our analysis also treats any feasible polytopes, worst-case objectives, and centralized strategies. The approximation guarantees proved here, therefore, do not follow exclusively from the greedy analysis for submodular maximization, but require additional arguments tailored to the stochastic monitoring framework.

\newcommand{\bidir}[1]{\ensuremath{\tilde{#1}}}

\section{Preliminaries}
\label{sec:prelim}

Let $\Delta_V$ be the standard simplex defined as 
$\Delta_V := \left\{ \mu \in \mathbb{R}^V_+ : \sum_{v \in V} \mu(v) = 1 \right\}$. The stationary distribution $\mu_i$ of the $i$-th robot  clearly lies in $\Delta_V$. 
{Let $\mathcal{P}$ be a polyhedral set to which $\mu_i$ belongs.  Note that $\mathcal{P} \subseteq \Delta_V$.}
Depending on the application, one may impose lower and upper bounds on $\mu_i(v)$, i.e., ${\mathrm{lb}}_v \le \mu_i(v) \le \mathrm{ub}_v$, resulting in a more restricted polytope $\mathcal{P}$.

Under the heterogeneous model, one should, in principle, define for each robot $i$ a transition matrix $T_i$ governing its motion, such that $\mu_i$ is a stationary distribution of $T_i$. In other words, the following conditions must hold:
$\mu_i^\top T_i = \mu_i^\top$,
$\sum_{v \in V} T_i(u,v) = 1 \;\; \forall u \in V$, and 
$T_i(u,v) = 0 \;\; \text{if } (u,v) \notin E$. 
Nevertheless, we can avoid incorporating the transition matrix into the optimization problems under study, thanks to the following lemma. 

\begin{lemma}[Stationary realizability]
\label{lem:realizability-positive}
Let $G=(V,E)$ be a connected undirected graph and assume $(v,v)\in E$ for all $v\in V$ (self-loops are allowed). Fix any probability distribution $\mu$ on $V$ such that $\mu(v)>0$ for all $v\in V$. Then there exists a transition matrix $T$ supported on $E$ (i.e., $T(u,v)=0$ if $(u,v)\notin E$) satisfying:
\begin{enumerate}
    \item $\mu^\top T = \mu^\top$ (i.e., $\mu$ is stationary for $T$),
    \item The Markov chain with transition matrix $T$ is irreducible and aperiodic,
    \item Consequently, $\mu$ is the unique stationary distribution of $T$; and regardless of the initial distribution, the distribution at time $t$ will converge to $\mu$ as $t \rightarrow \infty$.
\end{enumerate}
\label{lem:markov_convergence}
\end{lemma}
\begin{proof}
The construction follows the classical Metropolis--Hastings method \cite{metropolis1953equation,hastings1970monte}, adapted to finite graphs.

\noindent\textbf{Step 1: Proposal kernel.} Define a transition matrix $Q$ supported on $E$ by choosing uniformly among the closed neighborhood:
\[
Q_{ij} =
\begin{cases}
\frac{1}{|\mathcal{N}[i]|} & \text{if } j \in \mathcal{N}[i],\\
0 & \text{otherwise.}
\end{cases}
\]
Since $(i,i)\in E$, we have $Q_{ii}>0$ for all $i$.
Because $G$ is connected, the chain induced by $Q$ is irreducible, and due to self-loops it is aperiodic.

\noindent\textbf{Step 2: Metropolis correction.}
We modify $Q$ to obtain a new transition matrix $T$ whose stationary distribution is $\mu$.
For $i\neq j$ with $Q_{ij}>0$, define
\(
T_{ij}
=
Q_{ij}
\cdot
\min\!\left\{1,\,
\frac{\mu_j Q_{ji}}{\mu_i Q_{ij}}
\right\},
\)
and set
\(
T_{ii} = 1 - \sum_{j\neq i} T_{ij}.
\)
By construction, $T_{ij}\ge 0$ and each row sums to one, hence $T$ is a valid transition matrix supported on $E$.

\noindent\textbf{Step 3: Stationarity via detailed balance.} For every pair $i\neq j$ with $Q_{ij}>0$, the above definition ensures the detailed balance relations
\(
\mu(i) T_{ij} = \mu(j) T_{ji}.
\)
Summing over $i$ yields
\(
\sum_{i\in V} \mu(i) T_{ij} = \mu(j),
\)
and therefore $\mu^\top T = \mu^\top$. Thus $\mu$ is stationary for $T$.

\noindent\textbf{Step 4: Irreducibility and aperiodicity.} Since $G$ is connected and $\mu(v)>0$ for all $v$, every edge in $E$ receives strictly positive transition probability under $T$. 
Hence every vertex can be reached from every other vertex with positive probability, so the chain is irreducible.
Moreover, $T_{ii} \ge Q_{ii} > 0$ for all $i$, implying that every state has period one. 
Thus the chain is aperiodic.

\noindent\textbf{Step 5: Uniqueness.} For a finite irreducible Markov chain, the stationary distribution is unique (\cite[Theorem 1.7]{levin2017markov}). 
Therefore $\mu$ is the unique stationary distribution of $T$.
\end{proof}

In some cases, the mobility pattern of a robot may be further constrained. For instance, one can impose linear constraints on the transition probabilities, 
such as lower and upper bounds of the form ${\mathrm{lb}}_{u,v} \le T_i(u,v) \le \mathrm{ub}_{u,v}$. Nevertheless, we can again bypass the matrix $T_i$ and avoid the nonlinearities arising from the constraint $\mu_i^\top T_i = \mu_i^\top$ by considering the flow variables $f_i(u,v) := \mu_i(u) T_i(u,v)$, and 
adding flow conservation  constraints. 
This leads to the following polyhedral set involving $\mu_i$ and  $f_i$ variables:
\[
\mathcal{P}(\mu_i,f_i)
=
\left\{
(\mu_i,f_i):
\begin{aligned}
&\sum_{v\in V} f_i(u,v) = \mu_i(u) && \forall u\in V,\\
&\sum_{u\in V} f_i(u,v) = \mu_i(v) && \forall v\in V,\\
&f_i(u,v)=0 && \forall (u,v)\notin E,\\
&f_i(u,v)\ge 0 && \forall (u,v)\in E,\\
&\mathrm{lb}_{u,v}\mu_i(u)\le f_i(u,v)\le \mathrm{ub}_{u,v}\mu_i(u)
&& \forall (u,v)\in E,\\
&\mu_i\in\Delta_V,\\
&\mathrm{lb}_v\le \mu_i(v)\le \mathrm{ub}_v && \forall v\in V.
\end{aligned}
\right\}
\]
The polytope $\mathcal{P}$ can be viewed as the projection of $\mathcal{P}(\mu_i, f_i)$ onto the space of variables $\mu_i$. Observe that the transition matrix $T_i$ can be readily reconstructed from $(\mu_i, f_i)$ by defining $T_i(u,v) = \frac{f_i(u,v)}{\mu_i(u)}$ whenever $\mu_i(u) \neq 0$. Consequently, in the remainder of the paper, we omit explicit reference to the transition matrices and focus exclusively on stationary distributions $\mu_i$ that lie within a given polytope $\mathcal{P}$.


\section{Independent Homogeneous Monitoring}
\label{sec:hom}
In this section we study the {independent homogeneous} setting.
As introduced in Section \ref{sec:intro}, all $r$ robots execute the same homogeneous Markov chain on $G$ and, consequently, share a common stationary distribution $\mu \in \mathcal{P}$.
Robots move independently of one another, and coverage is evaluated in steady state.  
The coverage of a vertex $v$ is then simply given by $\cov(v, \mu,...,\mu) = 1-\left(1 - \sum_{u \in \mathcal{N}[v]} \mu(u)\right)^r = 1 - (1-\cov(v,\mu))^r  = \phi_r(\cov(v,\mu) )$, we define \(\phi_r(p):=1-(1-p)^r\).

First, we consider the average-case objective function, where the goal is to maximize the function
$F_{\mathrm{hom}}(r,\mu) := \sum_{v \in V} \phi_r\bigl(\cov(v, \mu)\bigr)$.
Observe that $F_{\mathrm{hom}}$ is concave in $\mu$, as it is a sum of concave functions. Consequently, maximizing $F_{\mathrm{hom}}$ amounts to solving a concave optimization problem, which can therefore be carried out efficiently. The optimal objective value will be denoted by $F^*_{\mathrm{hom}}(r)$.

Under the worst-case coverage function, defined for any $\mu \in \mathcal{P}$ as:
$W_{\mathrm{hom}}(r, \mu) := \min_{v \in V} \phi_r(\cov(v, \mu))$.
The function $W_{\mathrm{hom}}$ is also concave in $\mu$, as it is defined as the pointwise minimum of concave functions and can, therefore, be efficiently maximized.  
In fact, 
maximizing $W_{\mathrm{hom}}$ is equivalent to minimizing $\max_{v \in V} \left(1 - \sum_{u \in \mathcal{N}[v]} \mu(u)\right)^r$, which reduces to the minimization of $\max_{v \in V} \left(1 - \sum_{u \in \mathcal{N}[v]} \mu(u)\right)$. This, in turn, is equivalent to maximizing the minimum coverage, leading to the following linear program:
\begin{equation}
\label{eq:lp-homogeneous-prelim}
\begin{aligned}
\max_{\mu,\alpha} \qquad & \alpha\\
\text{s.t.}\qquad
& \sum_{u\in N[v]}\mu(u) \geq \alpha 
\qquad \forall v\in V,\\
& \mu\in \mathcal{P}.
\end{aligned}
\end{equation}
The optimal value $W^*_{\mathrm{hom}}( r)$ can then be efficiently computed by solving \eqref{eq:lp-homogeneous-prelim}.

\noindent We now show that $W^*_{\mathrm{hom}}(r)$ and $F^*_{\mathrm{hom}}(r)$ are star-concave in $r$, so $W^*_{\mathrm{hom}}(r)/r$ and $F^*_{\mathrm{hom}}(r)/r$ are non-increasing in $r$, implying diminishing marginal gains from additional robots.
\begin{proposition}
\label{prop:hom_star_concavity}
For every homogeneous monitoring instance and integer \(r\geq 1\),
the optimal average- and worst-case values satisfy the decreasing-ratio
property:
\[
\frac{F_{\hom}^*(r+1)}{r+1}
\leq
\frac{F_{\hom}^*(r)}{r},
\qquad
\frac{W_{\hom}^*(r+1)}{r+1}
\leq
\frac{W_{\hom}^*(r)}{r}.
\]
\end{proposition}
\begin{proof}
Let us first prove that $\frac{\phi_r(p)}{r} \geq \frac{\phi_{r+1}(p)}{r+1}$ for $p \in [0,1]$. This is equivalent to showing that
$g(p) = (r+1)(1 - (1-p)^r) - r(1 - (1-p)^{r+1})$
is non-negative for $p \in [0,1]$. We note that $g(p)$ is a non-decreasing function since its derivative is given by
$g'(p) = r(r+1)p(1-p)^{r-1}$
which is non-negative on the interval $[0,1]$. Thus, $g(p) \geq g(0) = 0$. \\
Let $\mu^*_{r+1}$ be a distribution such that $F^*_{\mathrm{hom}}( r+1) = F_{\mathrm{hom}}(r+1, \mu^*_{r+1})$. Using the previously established inequality, we deduce that for each $v \in V$,
$\phi_r\bigl({\cov}(v, \mu^*_{r+1})\bigr) \geq \frac{r}{r+1} \, \phi_{r+1}\bigl({\cov}(v, \mu^*_{r+1})\bigr)$. Summing up over all vertices leads to $F_{\mathrm{hom}}(r, \mu^*_{r+1}) \geq  \frac{r}{r+1} F^*_{\mathrm{hom}}( r+1) $. Hence, $F^*_{\mathrm{hom}}( r) \geq  \frac{r}{r+1} F^*_{\mathrm{hom}}( r+1) $. \\
Similarly, one can show that
$\frac{W^*_{\mathrm{hom}}( r+1)}{r+1} \leq \frac{W^*_{\mathrm{hom}}( r)}{r}$
by considering a distribution $\mu^*_{r+1}$ such that
$W^*_{\mathrm{hom}}( r+1) = W_{\mathrm{hom}}(r+1, \mu^*_{r+1})$.
Using the same inequality as before, we obtain for all $v \in V$,
$\phi_r\bigl({\cov}(v, \mu^*_{r+1})\bigr)
\geq \frac{r}{r+1}\, \phi_{r+1}\bigl({\cov}(v, \mu^*_{r+1})\bigr)$.
Taking the minimum over $v$ yields
$W_{\mathrm{hom}}(r, \mu^*_{r+1}) \geq \frac{r}{r+1} W^*_{\mathrm{hom}}( r+1)$,
which implies
$W^*_{\mathrm{hom}}( r) \geq \frac{r}{r+1} W^*_{\mathrm{hom}}( r+1)$.
\end{proof}
The decreasing-ratio property in Proposition~\ref{prop:hom_star_concavity} does not imply concavity of the optimal value as a function of $r$. 
For every fixed feasible distribution $\mu \in \mathcal P$, the functions
\(
r \mapsto F_{\mathrm{hom}}(r,\mu)\) and  \(
r \mapsto W_{\mathrm{hom}}(r,\mu)
\)
are concave in $r$. However, this does not imply concavity of the optimal values
\(
F_{\mathrm{hom}}^*(r)\)
and
\(
W_{\mathrm{hom}}^*(r)
\). See Example~\ref{ex:f-hom-is-not-concave}.

Although the homogeneous average- and worst-case problems are solvable in polynomial time (via convex optimization and linear programming), it is useful to identify graph classes with explicit optimizers. The key observation is that the homogeneous objectives depend only on the orbit structure under $\mathrm{Aut}(G)$, allowing a reduction from one variable per vertex to one per orbit. Throughout the section, we assume $\mathcal{P}=\Delta_V$ and, by slight abuse of notation, write $F_{\mathrm{hom}}(x)$ and $W_{\mathrm{hom}}(x)$ even when $x$ is not a probability distribution.

\begin{proposition}
\label{thm:orbit-symmetric}
Let $O_1,\dots,O_k$ be the vertex-orbits of $G$ under $\mathrm{Aut}(G)$. Then both homogeneous objectives admit an optimal $\mu^*\in \mathcal{P}= \Delta_V$ constant on orbits, i.e., $\mu^*(u)=x_i$ for $u\in O_i$, with $\sum_{i=1}^k |O_i|x_i=1$. Let $m_{ij}=|\mathcal{N}[v]\cap O_j|$ for $v\in O_i$ and $p_i(x)=\sum_{j=1}^k m_{ij}x_j$. Then $F_{\hom}(x)=\sum_{i=1}^k |O_i|\,\phi_r\bigl(p_i(x)\bigr)$ and $W_{\hom}(x)=\min_{1\le i\le k}\phi_r\bigl(p_i(x)\bigr)$, subject to $x_i\ge 0$ and $\sum_{i=1}^k |O_i|x_i=1$.\\
We defer the proof to Appendix~\ref{appendix:section:hom}.
\end{proposition}
Let $\gamma_f$ denote the fractional domination number of the graph, defined as follows. 
\begin{equation}
\label{eq:frac_dom}
\begin{aligned}
\gamma_f = \min_{x} \qquad & \sum_{u \in V} x(u) \\
\text{s.t.}\qquad
 \sum_{u\in N[v]} x(u) \geq 1,  
\quad  & x(v) \geq 0, \quad \forall v\in V.
\end{aligned}
\end{equation}
\begin{proposition}
\label{pro:homfrac}
Suppose that $\mathcal{P}=\Delta_V$. Then \( W^*_{\hhom}(r)=\phi_r(1/\gamma_f). \)
\end{proposition}
\begin{proof}
Let $(\mu^*,\alpha^*)$ be an optimal solution of~\eqref{eq:lp-homogeneous-prelim} with $\mathcal P=\Delta_V$. Setting \( x(u)=\frac{\mu^*(u)}{\alpha^*} \) gives a feasible solution to~\eqref{eq:frac_dom}, so $\gamma_f\le 1/\alpha^*$. Conversely, if $x^*$ is an optimal solution of~\eqref{eq:frac_dom}, then \( \mu(u)=\frac{x^*(u)}{\gamma_f}\) defines a feasible solution to~\eqref{eq:lp-homogeneous-prelim}, hence $\alpha^*\ge 1/\gamma_f$. Therefore, $\alpha^*=1/\gamma_f$, and the formula for $W^*_{\hom}$ follows.
\end{proof}
Using Proposition \ref{pro:homfrac} and results on fractional domination, one obtains
$W^*_{\mathrm{hom}}(r)=\phi_r\!\left(\frac{d+1}{|V|}\right)$ for vertex-transitive graphs with $d=\frac{2|E|}{|V|}$,  
$W^*_{\mathrm{hom}}(r)=\phi_r\!\left(\frac{1}{\left\lfloor (2^{l+1}+3)/7 \right\rfloor}\right)$ for the perfect binary tree  with $l$ levels, and $W^*_{\mathrm{hom}}(r)=\phi_r\!\left(\frac{ab-1}{2ab-a-b}\right)$ for the complete bipartite graph $K_{a,b}$ with $ab>1$. For $G = K_{a,b}$, we show in Appendix~\ref{appendix:section:hom} that $
    F_{\hhom}^*(r) = 1+\max(a,b)$ if $r=1$ and  $
    F_{\hhom}^*(r)=
     a+b -  \frac{(a-1)^r (b-1)^r}{\left( a (b-1)^{\frac{r}{r-1}} + b (a-1)^{\frac{r}{r-1}}  \right)^{r-1}}$ if $r >1$.

\section{Independent Heterogeneous Monitoring}
\label{sec:het}

We now consider the case where the $r$ robots may follow heterogeneous mobility patterns. As introduced in Section~\ref{sec:intro}, the coverage of a vertex $v \in V$ under mobility patterns $\mu_1, \dots, \mu_r$ is defined as
$\cov(v,\mu_1,\dots,\mu_r) = 1 - \prod_{i=1}^r \bigl(1 - {\cov}(v,\mu_i)\bigr)$.
In the average-case setting, the objective is to maximize
$F_{\mathrm{het}}(r,\mu_1,\dots,\mu_r) = \sum_{v \in V} {\cov}(v,\mu_1,\dots,\mu_r)$,
subject to $\mu_i \in \mathcal{P}$ for all $i = 1,\dots,r$.
In the worst-case setting, the objective is instead to maximize
$W_{\mathrm{het}}(r,\mu_1,\dots,\mu_r) = \min_{v \in V} {\cov}(v,\mu_1,\dots,\mu_r)$.
The corresponding optimal objective values are denoted by $F^*_{\mathrm{het}}(r)$ and $W^*_{\mathrm{het}}(r)$, respectively.

\subsection{Hardness of the Heterogeneous Model}
The first result states that it is NP-hard to maximize $F_{\mathrm{het}}(r,\mu_1,\dots,\mu_r)$ even when $r=2$. 

\begin{thm}
\label{th:r=2}
Computing \(F_{\mathrm{het}}^*(r)\) is NP-hard, even when the number of robots is \(r=2\), and the feasible set \(\mathcal P\) is specified by box constraints with strictly positive lower bounds together with the normalization condition \(\sum_{u\in V}\mu(u)=1\).
\end{thm}

\begin{proof}
We establish hardness by providing a reduction from the \textsc{Densest $k$-Subgraph} problem. In its decision variant, given a graph $H = (V(H), E(H))$ and two integers $k$ and $t$, the objective is to determine whether there exists a subset $S \subseteq V(H)$ such that
$|S| = k \quad \text{and} \quad |E(S)| \geq t$
where $E(S)$ denotes the set of edges in the subgraph induced by $S$. The \textsc{Densest $k$-Subgraph} problem generalizes the \textsc{Clique} problem. \\
Given an instance of the \textsc{Densest $k$-Subgraph} problem, we build an instance of our problem as follows. Consider the graph $G = (V(G), E(G))$ where the vertex set is defined as the union of five disjoint stable sets: $V(G) = V(H) \cup E(H) \cup E'(H) \cup X \cup X'$. The set $V(H)$ corresponds to the vertices of the original graph $H$, while $E(H)$ and its disjoint copy $E'(H)$ contain vertices representing the edges of $H$. Specifically, for each edge $\{u, v\} \in E(H)$, we denote its corresponding vertex in $E(H)$ as $x_{u,v}$ and its corresponding vertex in $E'(H)$ as $x'_{u,v}$. The remaining sets $X$ and $X'$ are auxiliary vertex sets of cardinality $|X| = |X'| = (tk)^2$ (see Figure~\ref{fig:hardness-heterogeneous-construction}). By construction, each of these five constituent sets forms an independent set in $G$. The edge set $E(G)$ is defined as follows. For each edge $\{u,v\} \in E(H)$, the graph $G$ contains the edges $\{x_{u,v}, u\}$ and $\{x_{u,v}, v\}$. Moreover, each vertex $x_{u,v}$ is adjacent to all vertices in  $E'(H) \setminus \{x'_{u,v}\}$. 
We also include all possible edges between $V(H)$ and $E'(H)$.
Finally, every vertex in $X$ is adjacent to every vertex in $V(H)$, and every vertex in $X'$ is adjacent to every vertex in $E'(H)$. \\
Let \(n:=|V(G)|\) and \(\gamma:=\frac{1}{kt}.\) Choose a rational number \(\eta>0\) such that
\(4n^2\eta<\gamma,\) equivalently \(\eta<\frac{1}{4n^2kt}.\) The polytope \(\mathcal P\) is defined by assigning a small background mass \(\eta\) to every vertex. 
The remaining probability mass \(1-n\eta\) can be distributed only over \(V(H)\cup E'(H)\), with each vertex of \(V(H)\) receiving at most \(\frac{1-n\eta}{k}\) additional mass and each vertex of \(E'(H)\) receiving at most \(\frac{1-n\eta}{t}\) additional mass.
 More formally:
\[
\resizebox{\textwidth}{!}{$
\mathcal P
=
\left\{
\mu \in \mathbb{R}^{|V(G)|}
\;:\;
\begin{alignedat}{4}
&\sum_{u\in V(G)} \mu(u)=1,
&\quad
&\eta \le \mu(x'_{u,v}) \le \eta+\frac{1-n\eta}{t}
&\quad
&\forall x'_{u,v}\in E'(H),\\[0.3em]
&\eta \le \mu(v) \le \eta+\frac{1-n\eta}{k}
&\quad
&\forall v\in V(H),
&\quad
&\mu(w)=\eta
&\quad
&\forall w\in X,\\[0.3em]
&\mu(w')=\eta
&\quad
&\forall w'\in X',
&\quad
&\mu(x_{u,v})=\eta
&\quad
&\forall x_{u,v}\in E(H).
\end{alignedat}
\right\}
$}
\]
Let
\(
T_0:=|V(G)|-\left(1-\frac{2}{k}\right)\)
and
\(
T_\eta:=T_0-\frac{\gamma}{2}.
\)
We show that, if the \textsc{Densest \(k\)-Subgraph} instance is a YES instance, then \(F_{\mathrm{het}}^*(r=2)>T_\eta\), while if it is a NO instance, then \(F_{\mathrm{het}}^*(r=2)<T_\eta\).
\subparagraph{$(\mathbf{\Leftarrow})$}
Assume that there exists a set \(S\subseteq V(H)\) such that \(|S|=k\) and \(|E_H(S)|\ge t\). We first define two normalized distributions \(\bar{\mu}_1\) and \(\bar{\mu}_2\), and then add the background mass required by the definition of \(\mathcal P\).
Define \(\bar{\mu}_1\) by setting
\(
\bar{\mu}_1(u)=\frac{1}{k}\quad \forall u\in S,
\)
and
\(
\bar{\mu}_1(u)=0
\)
for
\(
u\in V(G)\setminus S.
\)
Since \(|E_H(S)|\ge t\), choose \(t\) edges of \(H\) induced by \(S\). For each selected edge \(\{u,v\}\), put mass \(1/t\) on the corresponding vertex \(x'_{u,v}\in E'(H)\), and put zero mass elsewhere. This defines \(\bar{\mu}_2\).
Now define the actual feasible distributions
\(
\mu_i(u):=\eta+(1-n\eta)\bar{\mu}_i(u)
\)
for
\(
u\in V(G),
\)
\(
i=1,2.
\)
By the definition of \(\mathcal P\), both \(\mu_1\) and \(\mu_2\) belong to \(\mathcal P\). In particular, both distributions assign positive mass to every vertex.
Let us first compute the value of the normalized pair \((\bar{\mu}_1,\bar{\mu}_2)\). The same calculation yields
\[
F_{\mathrm{het}}(2,\bar{\mu}_1,\bar{\mu}_2)
=
|V(G)|-\left(1-\frac{2}{k}\right)
=
T_0.
\]
Indeed, all vertices in \(V(H)\cup E'(H)\cup X\cup X'\) are fully monitored. For a vertex \(x_{u,v}\in E(H)\), we have
\[
\operatorname{cov}(x_{u,v},\bar{\mu}_1)
=
\frac{|S\cap\{u,v\}|}{k},
\qquad
\operatorname{cov}(x_{u,v},\bar{\mu}_2)
=
1-\bar{\mu}_2(x'_{u,v}).
\]
The only vertices of \(E(H)\) that are not fully monitored are the \(t\) vertices corresponding to the selected induced edges of \(S\). For each such vertex, \(|S\cap\{u,v\}|=2\) and \(\bar{\mu}_2(x'_{u,v})=1/t\), so its uncovered probability is \(\left(1-\frac{2}{k}\right)\frac{1}{t}.\)
Thus the total uncovered probability is \(t\cdot (1-2/k)/t=1-2/k\), proving the displayed equality.

It remains to compare the normalized pair with the actual feasible pair. For any vertex \(v\in V(G)\) and each \(i\in\{1,2\}\),
\(
\left|
\operatorname{cov}(v,\mu_i)
-
\operatorname{cov}(v,\bar{\mu}_i)
\right|
\le n\eta,
\)
because \(\mu_i-\bar{\mu}_i=\eta\mathbf{1}-n\eta\,\bar{\mu}_i\), and the closed neighborhood of \(v\) has size at most \(n\). Since the function
\(
(p_1,p_2)\mapsto 1-(1-p_1)(1-p_2)
\)
is \(1\)-Lipschitz in each coordinate on \([0,1]^2\), the monitoring probability of each vertex changes by at most \(2n\eta\). Summing over all \(n\) vertices gives
\(
\left|
F_{\mathrm{het}}(2,\mu_1,\mu_2)
-
F_{\mathrm{het}}(2,\bar{\mu}_1,\bar{\mu}_2)
\right|
\le 2n^2\eta.
\)
Therefore,
\(
F_{\mathrm{het}}(2,\mu_1,\mu_2)
\ge
T_0-2n^2\eta.
\)
By our choice of \(\eta\), we have \(4n^2\eta<\gamma\), and hence \(2n^2\eta<\gamma/2\). Thus
\(
F_{\mathrm{het}}(2,\mu_1,\mu_2)
>
T_0-\frac{\gamma}{2}
=
T_\eta.
\)
Consequently, \(F_{\mathrm{het}}^*(2)>T_\eta\) in the YES case.
\subparagraph{$(\mathbf{\Rightarrow})$}
Assume that \(F_{\mathrm{het}}^*(2)\ge T_\eta\). Since \(F_{\mathrm{het}}\) is bilinear, its maximum over \(\mathcal P^2\) is attained at a pair of extreme points. Hence, there exist extreme points \(\mu_1,\mu_2\in\mathcal P\) such that \(F_{\mathrm{het}}(2,\mu_1,\mu_2)\ge T_\eta\).
For each \(i\in\{1,2\}\), define the normalized distribution \(\bar{\mu}_i\) by
\(\bar{\mu}_i(u):=\frac{\mu_i(u)-\eta}{1-n\eta}\) for each  \(u\in V(G)\).
By the definition of \(\mathcal P\), \(\bar{\mu}_i\) is supported on \(V(H)\cup E'(H)\), satisfies \(\bar{\mu}_i(v)\le 1/k\) for \(v\in V(H)\) and \(\bar{\mu}_i(x'_{u,v})\le 1/t\) for \(x'_{u,v}\in E'(H)\), and has total mass one. Moreover, since the map \(\bar{\mu}\mapsto \eta\mathbf 1+(1-n\eta)\bar{\mu}\) is an affine bijection, the distributions \(\bar{\mu}_1,\bar{\mu}_2\) are extreme points of the corresponding normalized polytope.

As in the previous direction, for each \(i\in\{1,2\}\) and each vertex \(v\), we have \(\left|\operatorname{cov}(v,\mu_i)-\operatorname{cov}(v,\bar{\mu}_i)\right|\le n\eta\). Therefore,
\(
\left|
F_{\mathrm{het}}(2,\mu_1,\mu_2)
-
F_{\mathrm{het}}(2,\bar{\mu}_1,\bar{\mu}_2)
\right|
\le 2n^2\eta.
\)
Since \(4n^2\eta<\gamma\), we have \(2n^2\eta<\gamma/2\). Hence
\begin{equation}
\label{eq:perturbed-lower-bound}
F_{\mathrm{het}}(2,\bar{\mu}_1,\bar{\mu}_2)
\ge
F_{\mathrm{het}}(2,\mu_1,\mu_2)-2n^2\eta
>
T_\eta-\frac{\gamma}{2}
=
T_0-\gamma.
\end{equation}
We now show that this implies the existence of a \(k\)-vertex set inducing at least \(t\) edges in \(H\).

Let \(\bar{\mu}\) be an extreme point of the normalized polytope. Write \(\bar{\mu}(E'(H))=a/t\) and \(\bar{\mu}(V(H))=b/k\). Since \(\bar{\mu}\) is supported on \(V(H)\cup E'(H)\), we have \(a/t+b/k=1\). We first observe that at least one of \(a\) and \(b\) must be an integer. Indeed, if both were fractional, then there would be a vertex \(w\in V(H)\) with \(0<\bar{\mu}(w)<1/k\), and a vertex \(x'_{u,v}\in E'(H)\) with \(0<\bar{\mu}(x'_{u,v})<1/t\). Shifting a sufficiently small amount of mass between these two non-tight coordinates would express \(\bar{\mu}\) as a convex combination of two distinct feasible points, contradicting extremality.

For \(i=1,2\), write \(\bar{\mu}_i(E'(H))=a_i/t\) and \(\bar{\mu}_i(V(H))=b_i/k\). Consider the vertices in \(X\). For every \(w\in X\), its normalized coverage by robot \(i\) is \(\operatorname{cov}(w,\bar{\mu}_i)=b_i/k\). Similarly, for every \(w'\in X'\), we have \(\operatorname{cov}(w',\bar{\mu}_i)=a_i/t\). Therefore,
\begin{align}
|V(G)|-F_{\mathrm{het}}(2,\bar{\mu}_1,\bar{\mu}_2)
&\ge
|X|\left(1-\frac{b_1}{k}\right)\left(1-\frac{b_2}{k}\right)
+
|X'|\left(1-\frac{a_1}{t}\right)\left(1-\frac{a_2}{t}\right) \nonumber\\
&=
(tk)^2\left(1-\frac{b_1}{k}\right)\left(1-\frac{b_2}{k}\right)
+
(tk)^2\left(1-\frac{a_1}{t}\right)\left(1-\frac{a_2}{t}\right).
\label{eq:positive-bound}
\end{align}

We next note that if \(1-b_i/k\neq 0\), then \(1-b_i/k\ge \min\{1/k,1/t\}\). Indeed, \(1-b_i/k=a_i/t\). If \(a_i\) is a positive integer, then \(a_i/t\ge 1/t\); otherwise \(b_i\) is an integer, and \(1-b_i/k\neq 0\) implies \(k-b_i\ge 1\), hence \(1-b_i/k\ge 1/k\). The same argument gives the analogous statement for \(1-a_i/t\).

Suppose now that \(\left(1-b_1/k\right)\left(1-b_2/k\right)\neq 0\). Then, by the preceding paragraph and \eqref{eq:positive-bound},
\[
|V(G)|-F_{\mathrm{het}}(2,\bar{\mu}_1,\bar{\mu}_2)
\ge
(tk)^2\min\{1/k,1/t\}^2
\ge 1.
\]
Thus \(F_{\mathrm{het}}(2,\bar{\mu}_1,\bar{\mu}_2)\le |V(G)|-1<T_0-\gamma\), since \(T_0-\gamma=|V(G)|-1+2/k-1/(kt)\). This contradicts \(F_{\mathrm{het}}(2,\bar{\mu}_1,\bar{\mu}_2)>T_0-\gamma\). Hence \(\left(1-b_1/k\right)\left(1-b_2/k\right)=0\). By the same argument applied to \(X'\), we also have \(\left(1-a_1/t\right)\left(1-a_2/t\right)=0\).

Since \(a_i/t+b_i/k=1\) for \(i=1,2\), the two equalities above imply, up to swapping the two robots, that \(b_1=k\), \(a_1=0\), \(b_2=0\), and \(a_2=t\). Therefore \(\bar{\mu}_1\) puts mass \(1/k\) on exactly \(k\) vertices of \(V(H)\), and \(\bar{\mu}_2\) puts mass \(1/t\) on exactly \(t\) vertices of \(E'(H)\). Define \(S:=\{u\in V(H):\bar{\mu}_1(u)=1/k\}\) where \(|S|=k\).

It remains to show that \(S\) induces at least \(t\) edges in \(H\). Suppose, toward a contradiction, that \(|E_H(S)|<t\). Since \(\bar{\mu}_2\) assigns mass \(1/t\) to exactly \(t\) vertices of \(E'(H)\), at least one of these vertices, say \(x'_{u,v}\), corresponds to an edge \(\{u,v\}\) with \(|S\cap\{u,v\}|\le 1\). For a selected vertex \(x'_{u,v}\), the corresponding vertex \(x_{u,v}\in E(H)\) has normalized uncovered probability
\(
\left(1-\frac{|S\cap\{u,v\}|}{k}\right)\frac1t.
\)
For selected edges contained in \(S\), this uncovered probability is \((1-2/k)/t\), while for at least one selected edge not contained in \(S\), it is at least \((1-1/k)/t\). Hence the total uncovered probability over the selected vertices of \(E(H)\) is at least
\[
(t-1)\frac{1-2/k}{t}+\frac{1-1/k}{t}
=
1-\frac{2}{k}+\frac{1}{kt}.
\]
Consequently,
$F_{\mathrm{het}}(2,\bar{\mu}_1,\bar{\mu}_2)
\le
|V(G)|-\left(1-\frac{2}{k}+\frac{1}{kt}\right)
=
T_0-\gamma$,
contradicting~\eqref{eq:perturbed-lower-bound}. 
Thus, from any solution with value at least \(T_\eta\), we can extract a \(k\)-vertex subset of \(H\) inducing at least \(t\) edges. This completes the reduction.
\end{proof}

\begin{figure}[t]
\centering
\begin{tikzpicture}[
    scale=0.85,
    x=1cm,y=1cm,
    every node/.style={font=\small},
    sample/.style={circle, draw, fill=white, inner sep=1.5pt},
    comp/.style={gray!35, thick},
    special/.style={black, very thick},
    missing/.style={red, dash dot, very thick}
]

\draw (-0.2,0.4) ellipse [x radius=0.90, y radius=1.70];
\draw (2.8,0.4) ellipse [x radius=0.90, y radius=1.70];
\draw (5.7,3.45) ellipse [x radius=1.70, y radius=0.90];
\draw (8.6,0.4) ellipse [x radius=0.90, y radius=1.70];
\draw (11.6,0.4) ellipse [x radius=0.90, y radius=1.70];

\node at (-0.2,2.45) {$X$};
\node at (2.8,2.45) {$V(H)$};
\node at (5.7,4.80) {$E(H)$};
\node at (8.6,2.45) {$E'(H)$};
\node at (11.6,2.45) {$X'$};

\node (x1) at (-0.2,1.55) {};
\node (x2) at (-0.2,0.95) {};
\node[text=gray!45] at (-0.2,0.35) {$\vdots$};
\node (x4) at (-0.2,-0.25) {};
\node (x5) at (-0.2,-0.85) {};

\node[sample,label=above:$u$] (v1) at (2.8,1.55) {};
\node                        (v2) at (2.8,0.95) {};
\node[text=gray!45] at (2.8,0.5) {$\vdots$};
\node[sample,label=above:$v$] (v4) at (2.8,-0.5) {};
\node                        (v5) at (2.8,-0.85) {};

\node (eh1) at (4.85,3.45) {};
\node[text=gray!45] at (5.25,3.45) {};
\node[sample,label=above:$x_{u,v}$] (xuv) at (5.70,3.45) {};
\node (eh3) at (6.35,3.45) {};

\node (e1) at (8.6,1.55) {};
\node (e2) at (8.6,0.95) {};
\node[text=gray!45] at (8.6,0.35) {$\vdots$};
\node[sample,label=below:$x'_{u,v}$] (xuvp) at (8.6,-0.15) {};
\node (e5) at (8.6,-0.85) {};

\node (xp1) at (11.6,1.55) {};
\node (xp2) at (11.6,0.95) {};
\node[text=gray!45] at (11.6,0.35) {$\vdots$};
\node (xp4) at (11.6,-0.25) {};

\draw[comp] (x1) -- (v1);
\draw[comp] (x1) -- (v2);

\draw[comp] (v1) -- (e1);
\draw[comp] (v1) -- (e2);
\draw[comp] (v2) -- (e1);

\draw[comp] (e1) -- (xp1);
\draw[comp] (e1) -- (xp2);
\draw[comp] (e2) -- (xp1);

\node[text=gray!45] at (10.1,0.35) {$\vdots$};
\node[text=gray!45] at (1.3,0.35) {$\vdots$};
\node[text=gray!45] at (5.7,0.35) {$\vdots$};

\draw[special] (xuv) -- (v1);
\draw[special] (xuv) -- (v4);

\draw[missing] (xuv) -- (xuvp);

\end{tikzpicture}
\caption{Graph \(G\) for the reduction from \textsc{Densest \(k\)-Subgraph}; red dash-dot: missing edges.}
\label{fig:hardness-heterogeneous-construction}
\end{figure}
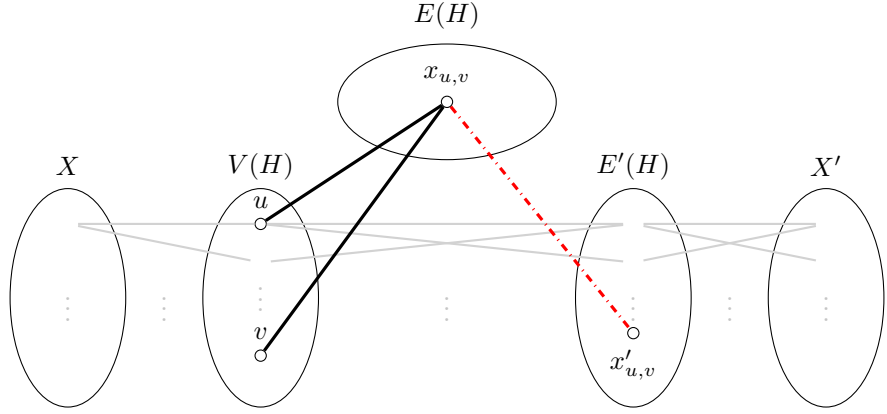
Although the hardness result of Theorem~\ref{th:r=2} holds even for \(r=2\), the value \(F^*_{\het}(r)\) can be computed efficiently when \(r\) is fixed and the polytope \(\mathcal P\) has polynomially many extreme points.

\begin{proposition}
\label{prop:exist-poly-fixed-r} 
\sloppy{If $r$ is bounded by a constant and the number of extreme points of $\mathcal{P}$ is polynomially bounded, then $F^*_{\mathrm{het}}(r)$ can be computed in polynomial time.}
\end{proposition}
\begin{proof}
Let $\sigma$ denote the number of extreme points of $\mathcal{P}$. 
Observe that the average coverage function
$F_{\mathrm{het}}(r,\mu_1,\dots,\mu_r) = \sum_{v \in V} \left( 1 - \prod_{i=1}^r \Bigl(1 - \sum_{u \in \mathcal{N}[v]} \mu_i(u) \Bigr) \right)$ 
is a multilinear function in $(\mu_1,\dots,\mu_r)$. Hence, its maximum over $\mathcal{P}^r$ is attained at extreme points of $\mathcal{P}$, i.e., each $\mu_i$ can be chosen as an extreme point of $\mathcal{P}$.
It follows that an optimal solution can be found by enumerating all $\sigma^r$ possible $r$-tuples of extreme points.
\end{proof}

We next show that the heterogeneous objectives remain hard to approximate. To this end, we restrict attention to a full-support polytope whose extreme points can be enumerated in polynomial time. For a rational number \(0<\eta<1/n\), define
\[
\Delta_V^\eta
:=
\left\{
\mu\in\mathbb R^{|V|}:
\sum_{v\in V}\mu(v)=1,\;
\mu(v)\geq \eta \text{ for every }v\in V
\right\}.
\]
Every \(\mu\in\Delta_V^\eta\) can be written uniquely as
\(
\mu=\eta\mathbf 1+(1-n\eta)\bar\mu,\)
\(
\bar\mu\in\Delta_V,
\)
where \(\mathbf 1\) denotes the all-ones vector. Hence \(\Delta_V^\eta\) is an affine copy of the standard simplex, and its extreme points are precisely
\(
\eta\mathbf 1+(1-n\eta)e_v,
\) for
\(
v\in V.
\)
In particular, \(\Delta_V^\eta\) has exactly \(n\) extreme points, which can be explicitly enumerated in polynomial time. It therefore satisfies the polynomial-extreme-point assumption of Proposition~\ref{prop:exist-poly-fixed-r}.

Moreover, every distribution in \(\Delta_V^\eta\) has full support. Thus, by Lemma~\ref{lem:realizability-positive}, every \(\mu\in\Delta_V^\eta\) is the unique stationary distribution of an irreducible and aperiodic Markov chain supported on \(G\).

As in the proof of Theorem \ref{th:r=2}, several of our hardness proofs utilize a perturbation argument. We briefly recall the main idea below. 
 Suppose that \(\mu_i=\eta\mathbf 1+(1-n\eta)\bar\mu_i\) for every \(i\in[r]\). Then, for every \(v\in V\) and \(i\in[r]\),
\(
\left|
\operatorname{cov}(v,\mu_i)
-
\operatorname{cov}(v,\bar\mu_i)
\right|
\leq n\eta.
\)
The function
\(
(p_1,\ldots,p_r)\longmapsto
1-\prod_{i=1}^r(1-p_i)
\)
is \(1\)-Lipschitz in each coordinate on \([0,1]^r\). Therefore, by changing the coordinates one at a time, replacing \((\bar\mu_1,\ldots,\bar\mu_r)\) with \((\mu_1,\ldots,\mu_r)\) changes the monitoring probability of each vertex by at most \(rn\eta\). Consequently, the average-case
objective changes by at most \(rn^2\eta\), whereas the worst-case objective changes by at most \(rn\eta\).
\begin{thm}
\label{thm:inapp-f-het}
For every \(\delta>0\), the heterogeneous average-case monitoring
problem admits no polynomial-time
\(\left(1-\frac{1}{e}+\delta\right)\)-approximation, unless
\(\mathrm{P}=\mathrm{NP}\), even when
\(\mathcal P=\Delta_V^\eta\) for a sufficiently small rational \(\eta>0\).
\end{thm}
\begin{proof}
It suffices to consider \(0<\delta<1/e\). Let \(\rho:=1-\frac{1}{e}+\delta,\) and suppose, for contradiction, that there exists a polynomial-time \(\rho\)-approximation algorithm for the heterogeneous average-case monitoring problem with \(\mathcal P=\Delta_V^\eta\). We show that this would yield a polynomial-time \(\left(1-\frac{1}{e}+\frac{\delta}{2}\right)\)-approximation for \textsc{Maximum Domination}, contradicting its known inapproximability~\cite{MIYANO}.

Consider an instance \((G,r)\) of \textsc{Maximum Domination}, let \(n:=|V(G)|\), and choose a rational number \(\eta>0\) satisfying
\(
\eta<\frac{1}{n}
\)
and
\(
rn^2\eta\leq\frac{\delta}{4}.
\)
Set \(\mathcal P:=\Delta_V^\eta\). Every
\(\mu_i\in\Delta_V^\eta\) can be written uniquely as
\(
\mu_i=\eta\mathbf 1+(1-n\eta)\bar\mu_i,
\)
where
\(
\bar\mu_i\in\Delta_V.
\)
By the perturbation bound established above, for every feasible tuple
\((\mu_1,\ldots,\mu_r)\in(\Delta_V^\eta)^r\) and its corresponding
normalized tuple \((\bar\mu_1,\ldots,\bar\mu_r)\), we have
\[
\left|
F_{\het}(r,\mu_1,\ldots,\mu_r)
-
F_{\het}(r,\bar\mu_1,\ldots,\bar\mu_r)
\right|
\leq rn^2\eta.
\]
Let \(\mathrm{OPT}_0\) denote the optimum over \((\Delta_V)^r\), and let \(\mathrm{OPT}_\eta\) denote the optimum over \((\Delta_V^\eta)^r\). The preceding bound and the affine correspondence between \(\Delta_V\) and \(\Delta_V^\eta\) imply that
\(
\left|\mathrm{OPT}_\eta-\mathrm{OPT}_0\right|
\leq rn^2\eta.
\)
The objective \(F_{\het}\) is affine in each marginal distribution when the remaining marginals are fixed. Hence a maximizer over \((\Delta_V)^r\) can be chosen so that every marginal is an extreme point of \(\Delta_V\), that is, a deterministic distribution \(e_v\). It follows that
\(
\mathrm{OPT}_0
=
\max_{\substack{S\subseteq V\\ |S|\leq r}}
|\N[S]|,
\)
which is precisely the optimum value of the given \textsc{Maximum Domination} instance.
Run the assumed approximation algorithm on the instance with \(\mathcal P=\Delta_V^\eta\), and let \((\mu_1,\ldots,\mu_r)\) be the returned solution. Then
\(
F_{\het}(r,\mu_1,\ldots,\mu_r)
\geq \rho\,\mathrm{OPT}_\eta.
\)
After normalizing the returned marginals as above, we obtain
\begin{align*}
F_{\het}(r,\bar\mu_1,\ldots,\bar\mu_r)
&\geq
F_{\het}(r,\mu_1,\ldots,\mu_r)-rn^2\eta\\
&\geq
\rho\,\mathrm{OPT}_\eta-rn^2\eta\\
&\geq
\rho\,\mathrm{OPT}_0-(\rho+1)rn^2\eta.
\end{align*}
Since \(\rho\leq 1\), \(rn^2\eta\leq\delta/4\), and
\(\mathrm{OPT}_0\geq 1\), it follows that
\[
F_{\het}(r,\bar\mu_1,\ldots,\bar\mu_r)
\geq
\left(\rho-\frac{\delta}{2}\right)\mathrm{OPT}_0
=
\left(1-\frac{1}{e}+\frac{\delta}{2}\right)\mathrm{OPT}_0.
\]
It remains to convert the normalized marginals into a deterministic
placement without decreasing the objective. Starting from
\((\bar\mu_1,\ldots,\bar\mu_r)\), process the marginals one at a time.
For a current tuple \((\nu_1,\ldots,\nu_r)\), the objective is affine
in \(\nu_i\) when all other marginals are fixed. Since
\(\nu_i=\sum_{v\in V}\nu_i(v)e_v\), we have
\(
F_{\het}(r,\nu_1,\ldots,\nu_r)
=
\sum_{v\in V}\nu_i(v)
F_{\het}(r,\nu_1,\ldots,\nu_{i-1},e_v,\nu_{i+1},\ldots,\nu_r).
\)
Hence, for some \(v_i\in V\), replacing \(\nu_i\) by \(e_{v_i}\)
does not decrease the objective. Such a vertex can be found by
evaluating the objective for all \(v\in V\). Repeating this procedure
for \(i=1,\ldots,r\) takes polynomial time and produces a
deterministic tuple \((e_{v_1},\ldots,e_{v_r})\).
Let \(S:=\{v_1,\ldots,v_r\}\). Then \(|S|\leq r\) and
\(
|\N[S]|
=
F_{\het}(r,e_{v_1},\ldots,e_{v_r})
\geq
\left(1-\frac{1}{e}+\frac{\delta}{2}\right)\mathrm{OPT}_0.
\)
This gives a polynomial-time \(\left(1-\frac{1}{e}+\frac{\delta}{2}\right)\)-approximation for \textsc{Maximum Domination}, a contradiction.
\end{proof}
We next establish the same inapproximability threshold for the heterogeneous worst-case monitoring problem.
\begin{thm}
\label{thm:inapp-het-worst}
For every \(\delta>0\), \(W^*_{\het}(r)\) cannot be approximated in polynomial time within a \(\left(1-\frac{1}{e}+\delta\right)\)-ratio, unless \(\mathrm{P}=\mathrm{NP}\), even when \(\mathcal P=\Delta_V^\eta\) for a sufficiently small \(\eta>0\).
\end{thm}

\begin{proof}
It suffices to consider \(0<\delta<1/e\). Let \(\rho:=1-\frac1e+\delta\), and set \(\beta:=1-\frac1e+\frac{\delta}{2}\). We reduce from the gap version of \textsc{Max \(k\)-Cover}~\cite{feige1998threshold}: it is NP-hard to distinguish between instances \((U,\mathcal S,k)\) in which some \(k\) sets cover all elements of \(U\), and instances in which every collection of \(k\) sets covers at most a \(\beta\)-fraction of \(U\). We assume, as in the standard formulation, that every element of \(U\) belongs to at least one set in \(\mathcal S\).

Given such an instance, construct a graph \(G=(V,E)\) as follows. For each set \(S_j\in\mathcal S\), create a set-vertex \(a_j\), and for each element \(e\in U\), create an element-vertex \(b_e\). Make all set-vertices pairwise adjacent, and connect \(a_j\) to \(b_e\) whenever \(e\in S_j\). There are no edges between element-vertices. Set the number of robots to \(r:=k\). Let \(n:=|V|\), and choose a rational number \(\eta>0\) such that \(0<\eta<1/n\) and \(\rho(1-rn\eta)>\beta\). Finally, set \(\mathcal P:=\Delta_V^\eta\).

We first present a simple property of this construction. For every deterministic configuration of the \(k\) robots, there is another deterministic configuration using only set-vertices that monitors at least the same set of element-vertices. Indeed, if a robot is placed on an element-vertex \(b_e\), then replacing it by any neighboring set-vertex \(a_j\) with \(e\in S_j\) preserves the monitoring of \(b_e\) and can only increase the set of monitored element-vertices. Consequently, in the NO case, every deterministic configuration monitors at most \(\beta|U|\) element-vertices.

Consider first the YES case. Let \(a_{j_1},\ldots,a_{j_k}\) be set-vertices corresponding to \(k\) sets whose union covers \(U\). In the standard simplex \(\Delta_V\), placing robot \(i\) deterministically on \(a_{j_i}\) monitors every element-vertex. It also monitors every set-vertex, since the set-vertices form a clique and at least one robot is placed in this clique. Thus the worst-case value over \(\Delta_V\) is \(1\). Replacing these deterministic distributions by their full-support perturbations in \(\Delta_V^\eta\) changes each vertex monitoring probability by at most \(rn\eta\). Therefore, in the YES case, \(W^*_{\het}(r)\ge 1-rn\eta\).

Now consider the NO case. Let \((\mu_1,\ldots,\mu_r)\in(\Delta_V^\eta)^r\) be any feasible tuple, and let \(\pi:=\mu_1\otimes\cdots\otimes\mu_r\) be the product distribution that it induces over deterministic robot configurations. For each \(e\in U\), define
\(
q_e:=\Pr_{s\sim\pi}\bigl[b_e\text{ is monitored in }s\bigr].
\)
Since every deterministic configuration in the support of \(\pi\) monitors at most \(\beta|U|\) element-vertices, we have
\[
\sum_{e\in U}q_e
=
\mathbb E_{s\sim\pi}
\bigl[
  \#\{\text{element-vertices monitored in }s\}
\bigr]
\le
\beta|U|.
\]
Hence some \(e\in U\) satisfies \(q_e\le\beta\). Since \(q_e\) is the monitoring probability of \(b_e\) under the tuple \((\mu_1,\ldots,\mu_r)\), it follows that
\(
W_{\het}(r,\mu_1,\ldots,\mu_r)\le q_e\le\beta.
\)
Since the tuple was arbitrary, \(W_{\het}^*(r)\le\beta\) in the NO case.

Suppose now that a polynomial-time \(\rho\)-approximation algorithm existed for the heterogeneous worst-case problem with \(\mathcal P=\Delta_V^\eta\). In the YES case, it would return a solution of value at least \(\rho(1-rn\eta)>\beta\), whereas in the NO case every feasible solution has value at most \(\beta\). Thus the algorithm would distinguish the YES and NO cases of the gap \textsc{Max \(k\)-Cover} instance, which is impossible unless \(\mathrm{P}=\mathrm{NP}\).
\end{proof}

\subsection{Approximation via the homogeneous problem}
\label{subsec:het-approx}

The heterogeneous model is strictly more expressive than the homogeneous one. A natural question is therefore to quantify the loss incurred by restricting attention to homogeneous strategies. The next result shows that solving the homogeneous problem and replicating the resulting distribution across all robots yields a \(\alpha_r =  1-\left(1-\frac{1}{r}\right)^r \)-approximation for the heterogeneous optimum.
We begin by stating the guarantee for the average-case objective. 
\begin{thm}
\label{thm:het-avg-alpha}
Let \((\mu_1^*,\dots,\mu_r^*)\in \mathcal{P}^r\) be an optimal heterogeneous solution for the average-case objective. If \(\widehat{\mu}\in\mathcal{P}\) is an optimal solution of the homogeneous problem and we replicate \(\widehat{\mu}\) across all \(r\) robots, then
\[
F^*_{\hom}(r)=F_{\hom}(r,\widehat{\mu})=F_{\het}(r,\widehat{\mu},\dots,\widehat{\mu})
\;\ge\;
\alpha_r F^*_{\het}(r)\ge \alpha_r F^*_{\hom}(r).
\]
Moreover, if \(\mathcal P\) has a polynomial-size rational linear description, then, for every \(\varepsilon>0\), convex optimization computes, in time polynomial in the input size and \(\log(1/\varepsilon)\), a homogeneous distribution whose replication across all \(r\) robots achieves a value of at least \((\alpha_r-\varepsilon)F_{\het}^*(r)\).
\end{thm}

\begin{proof}
Let
\(
\bar{\mu}:=\frac{1}{r}\sum_{i=1}^r \mu_i^*.
\)
Since $\mathcal{P}$ is convex, we have $\bar{\mu}\in \mathcal{P}$. By optimality of $\widehat{\mu}$ for the homogeneous problem,
\(
F_{\hhom}(r,\widehat{\mu})\ge F_{\hhom}(r,\bar{\mu}).
\)
It therefore suffices to compare the homogeneous value of $\bar{\mu}$ with the heterogeneous optimum. Fix a vertex $v\in V$, and write
$
x_i(v):=\sum_{u\in \N[v]}\mu_i^*(u)$ and 
$s_v:=\sum_{i=1}^r x_i(v)$.
Then
\(
\sum_{u\in \N[v]}\bar{\mu}(u)=\frac{s_v}{r}.
\)
Hence the homogeneous monitoring probability at $v$ under $\bar{\mu}$ is
\(
1-\left(1-\frac{s_v}{r}\right)^r.
\) \\
On the other hand, by considering the monitoring probability under the heterogeneous model and applying the union bound, we obtain
\begin{equation}
    1 - \prod_{i=1}^r \bigl(1 - x_i(v)\bigr) \le \min\{1, s_v\}\,.
    \label{eq:prod}
\end{equation}
The left-hand side represents the probability that $v$ is monitored, while $s_v$ denotes the sum, over all robots, of the probabilities of covering~$v$, which yields the bound by the union bound.

\begin{claim}
 $1-\left(1-\frac{s}{r}\right)^r \ge \alpha_r \min\{1,s\}$ for $s\in[0,r]$.  
 \label{claim:het-avg-alpha}
 \end{claim}
\begin{proof}
Consider $f(s)=1-(1-s/r)^r$ on $[0,r]$. If $0\le s\le 1$, then $f$ is concave, $f(0)=0$, and $f(1)=\alpha_r$, so
$f(s)\ge s\,f(1)=\alpha_r s=\alpha_r\min\{1,s\}$.
If $1\le s\le r$, then $f$ is increasing, and therefore
$f(s)\ge f(1)=\alpha_r=\alpha_r\min\{1,s\}$. Thus
$1-\left(1-\frac{s}{r}\right)^r\ge \alpha_r\min\{1,s\}$
$\forall s\in[0,r]$.
\end{proof}
Applying Claim~\ref{claim:het-avg-alpha} for  \(s=s_v\) and using \eqref{eq:prod} gives
\[
1-\left(1-\sum_{u\in \N[v]}\bar{\mu}(u)\right)^r
=
1-\left(1-\frac{s_v}{r}\right)^r
\ge
\alpha_r \left(1-\prod_{i=1}^r(1-x_i(v))\right)\,.
\]
Summing over all \(v\in V\) yields 
$F_{\hhom}(r,\bar{\mu})\ge \alpha_r \,F_{\het}(r,\mu_1^*,\dots,\mu_r^*)
=
\alpha_r \,F^*_{het}(r)$.
Combining this with the optimality of \(\widehat{\mu}\) proves the first statement of the theorem.

For the computational statement, let \(\widehat{\mu}_\varepsilon\in\mathcal P\) be a homogeneous solution whose value is within additive error \(\varepsilon\) of the homogeneous optimum. Such a distribution can be computed by convex optimization in time polynomial in the input size and \(\log(1/\varepsilon)\). Replicating it across all \(r\) robots gives
\(
F_{\het}(r,\widehat{\mu}_\varepsilon,\dots,\widehat{\mu}_\varepsilon)
=
F_{\hom}(r,\widehat{\mu}_\varepsilon)
\ge
F_{\hom}^*(r)-\varepsilon
\ge
\alpha_r F_{\het}^*(r)-\varepsilon.
\)
Since \(F_{\het}^*(r)\ge1\) for every nonempty instance, we have
\(
\alpha_r F_{\het}^*(r)-\varepsilon
\ge
(\alpha_r-\varepsilon)F_{\het}^*(r).
\)
\end{proof}
The same argument applies to the worst-case objective, since the underlying comparison is pointwise in the vertex and therefore survives replacing the summation over \(v\) by a minimum.

\begin{corollary}
\label{cor:het-worst-alpha}
Let $\widehat{\mu} \in  \arg\max_{\mu\in\mathcal{P}} W_{\hhom}(r,\mu)
$. Then
\[
W^*_{\hhom}(r) = 
W_{\hhom}(r,\widehat{\mu}) = W_{\het}(r,\widehat{\mu},...,\widehat{\mu})
\;\ge\;
\alpha_r  \,W^*_{\het}(r) \ge \alpha_r  \,W^*_{\hhom}(r).
\] 
\end{corollary}
The factor \( \left( 1-\left(1-\frac{1}{r}\right)^r\right)\) is tight under both objectives, as shown in Examples~\ref{ex:tight-connected} and~\ref{ex:tight-worst}.

\subsection{Star-concavity} 
As in the homogeneous case, $F^{*}_{\het}(r)$ is star-concave. In contrast, $W^{*}_{\het}(r)$ is not star-concave in general, as shown in Example~\ref{exp:concavity-het-worst}.
\begin{proposition}
\label{prop:starconcavity-het-f}
For every \(r\ge1\),
\(
\frac{F_{\het}^*(r+1)}{r+1}
\le
\frac{F_{\het}^*(r)}{r}.
\)
\end{proposition}

\begin{proof}
Let \(\mu_1^*,\ldots,\mu_{r+1}^*\in\mathcal P\) be an optimal heterogeneous solution for \(r+1\) robots, so that
\(
F_{\het}^*(r+1)
=
F_{\het}(r+1,\mu_1^*,\ldots,\mu_{r+1}^*).
\)
For each \(i\in[r+1]\), consider the feasible \(r\)-robot strategy obtained by removing robot \(i\), namely
\(
\bigl(\mu_j^*\bigr)_{j\in[r+1]\setminus\{i\}}.
\)
We show that the average value of these \(r+1\) strategies is at least
\(\frac{r}{r+1}F_{\het}^*(r+1)\).
Fix a vertex \(v\in V\). Independently sample a location \(S_j\sim\mu_j^*\) for each \(j\in[r+1]\), and let
\(
C_v(S)
:=
\{j\in[r+1]:S_j\in\N[v]\}
\)
be the set of robots covering \(v\) in the resulting configuration. Now choose an index \(D\) uniformly from \([r+1]\) and delete robot \(D\).

If \(C_v(S)=\varnothing\), then \(v\) is uncovered both before and after the deletion. If \(C_v(S)\neq\varnothing\), deleting one robot preserves the coverage of \(v\) with probability at least \(r/(r+1)\): when exactly one robot covers \(v\), coverage is lost only if that robot is deleted, while if at least two robots cover \(v\), deleting a single robot cannot destroy coverage. Therefore,
\(
\Pr_D\!\left[
C_v(S)\setminus\{D\}\neq\varnothing
\,\middle|\, S
\right]
\ge
\frac{r}{r+1}\,
\mathbf 1_{\{C_v(S)\neq\varnothing\}}.
\)
Taking expectation over \(S\) and then averaging over the deleted index gives
\[
\frac{1}{r+1}
\sum_{i=1}^{r+1}
\cov\!\left(
v,
(\mu_j^*)_{j\in[r+1]\setminus\{i\}}
\right)
\ge
\frac{r}{r+1}
\cov(v,\mu_1^*,\ldots,\mu_{r+1}^*).
\]
Summing this pointwise inequality over all \(v\in V\), we obtain
\[
\frac{1}{r+1}
\sum_{i=1}^{r+1}
F_{\het}\!\left(
r,
(\mu_j^*)_{j\in[r+1]\setminus\{i\}}
\right)
\ge
\frac{r}{r+1}F_{\het}^*(r+1).
\]
Consequently, there exists an index \(i^*\in[r+1]\) such that
\(
F_{\het}\!\left(
r,
(\mu_j^*)_{j\in[r+1]\setminus\{i^*\}}
\right)
\ge
\frac{r}{r+1}F_{\het}^*(r+1).
\)
Since the remaining \(r\) distributions constitute a feasible heterogeneous strategy,
\(
F_{\het}^*(r)
\ge
\frac{r}{r+1}F_{\het}^*(r+1).
\)
Dividing by \(r\) proves the claimed inequality.
\end{proof}
Note that the proof above cannot be replicated for the worst-case objective function: the summation over all vertices is a crucial component of the argument.  

\section{Centralized Monitoring}
\label{sec:centralized}

We now turn to the centralized model, in which robot locations need not be independent. At a steady-state observation time, the team configuration is represented by a vector \(s=(s_1,\ldots,s_r)\in V^r\), where \(s_i\) is the location of robot \(i\). A centralized strategy is a joint distribution \(\pi\in\Delta_{V^r}\) over such configurations. The marginal distribution of robot \(i\) is given by
\(
\mu_i^\pi(u)
:=
\sum_{\substack{s\in V^r\\s_i=u}}\pi(s),\)
for \( u\in V,
\)
and feasibility requires \(\mu_i^\pi\in\mathcal P\) for every \(i\in[r]\). Thus, the centralized model preserves the individual feasibility constraints imposed in the independent models while allowing correlations between robot locations. Such correlations can reduce redundant coverage and improve the overall monitoring performance.
For a vertex \(v\in V\), its monitoring probability under \(\pi\) is
\[
\cov_{\cent}(v,\pi)
:=
\Pr_{s\sim\pi}\!\bigl[\exists i\in[r]:s_i\in\N[v]\bigr]
=
1-\sum_{s\in(V\setminus\N[v])^r}\pi(s).
\]
Accordingly, the centralized average- and worst-case objectives are
\[
F_{\cent}(r,\pi)
:=
\sum_{v\in V}\cov_{\cent}(v,\pi),
\qquad
W_{\cent}(r,\pi)
:=
\min_{v\in V}\cov_{\cent}(v,\pi),
\]
with optimal values
\[
F_{\cent}^*(r)
:=
\max_{\substack{\pi\in\Delta_{V^r}\\
                 \mu_i^\pi\in\mathcal P\;\forall i\in[r]}}
F_{\cent}(r,\pi),
\qquad
W_{\cent}^*(r)
:=
\max_{\substack{\pi\in\Delta_{V^r}\\
                 \mu_i^\pi\in\mathcal P\;\forall i\in[r]}}
W_{\cent}(r,\pi).
\]
Both objectives admit linear programming formulations with \(n^r\) joint-probability variables: the average-case objective is linear in \(\pi\), while the worst-case objective can be linearized by introducing a variable \(t\) to be maximized and imposing \(\cov_{\cent}(v,\pi)\ge t\) for every \(v\in V\). Consequently, both problems can be solved in time polynomial in \(n^r\) and in the description size of \(\mathcal P\), and hence in polynomial time whenever \(r\) is fixed. When \(r\) is part of the input, however, both problems are NP-hard and, as shown in Section~\ref{subsec:cent-hardness}, cannot be approximated within \(1-\nicefrac{1}{e}+\delta\) for any constant \(\delta>0\), unless \(\mathrm{P}=\mathrm{NP}\).

We first develop in Section~\ref{subsec:cent-approx} a hierarchy of block-coordinated approximations. For a block size \(r'\in[r]\), write \(r=hr'+b\), where \(h=\lfloor r/r'\rfloor\) and \(0\le b<r'\). The construction coordinates the robots within blocks of size \(r'\), repeats the same block strategy independently across the \(h\) full blocks, and obtains the strategy of the remaining \(b\) robots by uniform subsampling. It achieves the approximation factor
\begin{equation}
\label{eq:block-factor}
\alpha_{r,r'}
:=
1-\left(1-\frac{r'}{r}\right)^h
\left(1-\frac{b}{r}\right).
\end{equation}
For \(r'=1\), the construction recovers independent homogeneous monitoring and the factor \(\alpha_r=1-(1-1/r)^r\); for \(r'=r\), it coincides with the exact centralized problem. Section~\ref{subsec:cent-hardness} establishes the matching asymptotic inapproximability threshold, and the final subsection proves that both \(F_{\cent}^*(r)/r\) and \(W_{\cent}^*(r)/r\) are nonincreasing in the number of robots.

\subsection{A hierarchy of block-coordinated approximations}
\label{subsec:cent-approx}

We now give a family of approximations that interpolates between independent
homogeneous monitoring and fully centralized monitoring. Fix an integer
\(r'\in[r]\), and write \(r=hr'+b\), where \(h=\lfloor r/r'\rfloor\) and
\(0\le b<r'\). The algorithm coordinates the robots within blocks of size
\(r'\), uses independent copies of the same strategy for the \(h\) full
blocks, and derives the strategy of the remaining \(b\) robots by uniformly
subsampling one full block.
For \(m\ge 1\), let
\(
\mathcal Q_m
:=
\bigl\{\pi\in\Delta_{V^m}:\mu_i^\pi\in\mathcal P
\text{ for every }i\in[m]\bigr\}
\)
denote the set of feasible \(m\)-robot joint strategies. We first record the
subsampling property used throughout the construction.

\begin{lemma}
\label{lem:cent-uniform-subsampling}
Let \(\pi\in\mathcal Q_m\), let \(0\le q\le m\), and define
\(\operatorname{Sub}_q(\pi)\) as follows. Sample
\(S=(S_1,\ldots,S_m)\sim\pi\), choose
\(I=\{i_1<\cdots<i_q\}\) uniformly among the \(q\)-element subsets of
\([m]\), and return
\((S_{i_1},\ldots,S_{i_q})\). Then \(\operatorname{Sub}_q(\pi)\) is a
feasible \(q\)-robot strategy. Moreover, for every \(v\in V\),
\(
\cov_q\bigl(v,\operatorname{Sub}_q(\pi)\bigr)
\ge \frac{q}{m}\,\cov_m(v,\pi).
\)
For \(q=0\), the marginal condition is vacuous and the coverage probability
is defined to be zero.
\end{lemma}

\begin{proof}
Suppose first that \(q\ge 1\). For every \(\ell\in[q]\), the \(\ell\)-th
marginal of \(\operatorname{Sub}_q(\pi)\) is
\(
\sum_{j=\ell}^{m-q+\ell}
\frac{\binom{j-1}{\ell-1}\binom{m-j}{q-\ell}}{\binom{m}{q}}
\,\mu_j^\pi.
\)
The coefficients are nonnegative and sum to one: they are precisely the
probabilities that \(j\) is the \(\ell\)-th smallest index in the uniformly
chosen set \(I\). Hence this marginal belongs to \(\mathcal P\), because
\(\mathcal P\) is convex.
Fix \(v\in V\), and condition on a configuration
\(s=(s_1,\ldots,s_m)\). Let
\(c_v(s):=|\{i\in[m]:s_i\in\N[v]\}|\). If \(c_v(s)\ge 1\), the probability
that the sampled \(q\)-subconfiguration still covers \(v\) is
\(
1-\frac{\binom{m-c_v(s)}{q}}{\binom{m}{q}}
\ge
1-\frac{\binom{m-1}{q}}{\binom{m}{q}}
=\frac qm.
\)
Averaging over \(s\sim\pi\) proves the coverage inequality. The assertion for
\(q=0\) is immediate.
\end{proof}
Let \(\pi\in\mathcal Q_{r'}\), set \(\pi^{(b)}:=\operatorname{Sub}_b(\pi)\), and define the repeated-block strategy \(\mathcal R_{r,r'}(\pi):=\pi^{\otimes h}\otimes\pi^{(b)},\)
where the last factor is omitted when \(b=0\). Lemma~\ref{lem:cent-uniform-subsampling}
implies that \(\mathcal R_{r,r'}(\pi)\) is feasible. For a fixed vertex \(v\), put \(x_v:=\cov_{r'}(v,\pi)\) and \(y_v:=\cov_b(v,\pi^{(b)})\). Since the blocks are independent,
\(
\cov_{\cent}\bigl(v,\mathcal R_{r,r'}(\pi)\bigr)
=1-(1-x_v)^h(1-y_v).
\)
Let \(\lambda:=b/r'\). The subsampling lemma gives \(y_v\ge\lambda x_v\), and therefore
\begin{equation}
\label{eq:block-pointwise-surrogate}
\cov_{\cent}\bigl(v,\mathcal R_{r,r'}(\pi)\bigr)
\ge g_{h,\lambda}(x_v),
\qquad
g_{h,\lambda}(x):=1-(1-x)^h(1-\lambda x).
\end{equation}
A direct differentiation shows that \(g_{h,\lambda}\) is nondecreasing and concave on \([0,1]\). We consequently optimize the surrogate objectives
\begin{equation}
\label{eq:block-surrogates}
\begin{aligned}
\widetilde F^{\mathrm{block}}_{r'}(r,\pi)
:=\sum_{v\in V}g_{h,\lambda}\bigl(\cov_{r'}(v,\pi)\bigr) \quad \text{and} \quad
\widetilde W^{\mathrm{block}}_{r'}(r,\pi)
:=\min_{v\in V}g_{h,\lambda}\bigl(\cov_{r'}(v,\pi)\bigr).
\end{aligned}
\end{equation}
By~\eqref{eq:block-pointwise-surrogate}, the value attained by the resulting repeated-block strategy is at least the corresponding surrogate value.

Both programs in~\eqref{eq:block-surrogates} are convex-optimization problems in the standard sense of maximizing a concave objective over a polytope. In fact, because \(g_{h,\lambda}\) is nondecreasing, maximizing the worst-case surrogate is equivalent to the linear program
\begin{equation}
\label{eq:block-worst-lp}
\begin{aligned}
\max_{\pi,t}\quad &t\\
\text{s.t.}\quad
&\cov_{r'}(v,\pi)\ge t &&\forall v\in V,\\
&\pi\in\mathcal Q_{r'}.
\end{aligned}
\end{equation}
The block formulation uses \(n^{r'}\) joint-probability variables, compared with \(n^r\) variables in the exact centralized formulation.

\begin{thm}
\label{thm:cent-block-approx}
Let \(r,r'\in\mathbb Z_{>0}\) satisfy \(1\le r'\le r\). Define the integers \(h:=\left\lfloor\frac{r}{r'}\right\rfloor\) and \(b:=r-hr',\) so that \(r=hr'+b\) and \(0\le b<r'\), and set \(\lambda:=b/r'\). Let \(\hat\pi_F\) and \(\hat\pi_W\) maximize,
respectively, \(\widetilde F^{\mathrm{block}}_{r'}(r,\pi)\) and \(\widetilde W^{\mathrm{block}}_{r'}(r,\pi)\) over \(\pi\in\mathcal Q_{r'}\). Then
\[
F_{\cent}\bigl(r,\mathcal R_{r,r'}(\hat\pi_F)\bigr)
\ge
\alpha_{r,r'}F_{\cent}^*(r)
\quad\text{and}\quad
W_{\cent}\bigl(r,\mathcal R_{r,r'}(\hat\pi_W)\bigr)
\ge
\alpha_{r,r'}W_{\cent}^*(r).
\]
If \(r'\) is constant and \(\mathcal P\) has a polynomial-size rational linear description, then the worst-case strategy is computable in polynomial time by~\eqref{eq:block-worst-lp}.
Moreover, for every \(\varepsilon>0\), convex optimization computes, in time polynomial in the input size and \(\log(1/\varepsilon)\), an average-case strategy with value at least \((\alpha_{r,r'}-\varepsilon)F_{\cent}^*(r).
\)
\end{thm}

\begin{proof}
We first establish the comparison with an arbitrary feasible
centralized strategy \(\Pi\in\Delta_{V^r}\). Define
\(\bar\pi:=\operatorname{Sub}_{r'}(\Pi)\). By
Lemma~\ref{lem:cent-uniform-subsampling},
\(\bar\pi\in\mathcal Q_{r'}\) and, for every \(v\in V\),
\begin{equation}
\label{eq:block-sampled-coverage}
x_v:=\cov_{r'}(v,\bar\pi)
\ge
\frac{r'}{r}\cov_{\cent}(v,\Pi).
\end{equation}

\begin{claim}
\label{claim:block-scalar-bound}
For every \(x\in[0,1]\),
\(
g_{h,\lambda}(x)
\ge
\alpha_{r,r'}\min\{1,(h+\lambda)x\}.
\)
\end{claim}

\begin{proof}
Since \(h+\lambda=r/r'\), we have
\(
g_{h,\lambda}\left(\frac{1}{h+\lambda}\right)
=
1-\left(1-\frac{r'}{r}\right)^h
  \left(1-\frac{b}{r}\right)
=
\alpha_{r,r'}.
\)
For \(0\le x\le 1/(h+\lambda)\), the concavity of
\(g_{h,\lambda}\), together with \(g_{h,\lambda}(0)=0\), gives
\[
g_{h,\lambda}(x)
\ge
(h+\lambda)x\,
g_{h,\lambda}\left(\frac{1}{h+\lambda}\right)
=
\alpha_{r,r'}(h+\lambda)x.
\]
For \(x\ge 1/(h+\lambda)\), the monotonicity of \(g_{h,\lambda}\)
gives
\(
g_{h,\lambda}(x)
\ge
g_{h,\lambda}\left(\frac{1}{h+\lambda}\right)
=
\alpha_{r,r'}.
\)
The two cases prove the claim.
\end{proof}
By~\eqref{eq:block-sampled-coverage} and \((h+\lambda)r'/r=1\), we have
\(
(h+\lambda)x_v
\ge
\cov_{\cent}(v,\Pi).
\)
Since \(\cov_{\cent}(v,\Pi)\le1\), applying
Claim~\ref{claim:block-scalar-bound} yields the pointwise comparison
\begin{equation}
\label{eq:block-pointwise-comparison}
g_{h,\lambda}\bigl(\cov_{r'}(v,\bar\pi)\bigr)
\ge
\alpha_{r,r'}\cov_{\cent}(v,\Pi)
\qquad
\forall v\in V.
\end{equation}
Apply this construction first to an optimal average-case centralized
strategy \(\Pi_F^*\), and set
\(
\bar\pi_F:=\operatorname{Sub}_{r'}(\Pi_F^*).
\)
Thus, \(\bar\pi_F\) is the corresponding subsampled \(r'\)-robot strategy. By the optimality of \(\hat\pi_F\), together with \eqref{eq:block-pointwise-surrogate} and \eqref{eq:block-pointwise-comparison}, we obtain
\begin{align*}
F_{\cent}\bigl(r,\mathcal R_{r,r'}(\hat\pi_F)\bigr)
&\ge
\widetilde F^{\mathrm{block}}_{r'}(r,\hat\pi_F)\\
&\ge
\widetilde F^{\mathrm{block}}_{r'}(r,\bar\pi_F)\\
&=
\sum_{v\in V}
g_{h,\lambda}\bigl(\cov_{r'}(v,\bar\pi_F)\bigr)\\
&\ge
\alpha_{r,r'}
\sum_{v\in V}\cov_{\cent}(v,\Pi_F^*)\\
&=
\alpha_{r,r'}F_{\cent}^*(r).
\end{align*}
Similarly, let \(\Pi_W^*\) be an optimal worst-case centralized strategy and set
\(\bar\pi_W:=\operatorname{Sub}_{r'}(\Pi_W^*).\)
Using the optimality of \(\hat\pi_W\) and the same pointwise comparison gives
\begin{align*}
W_{\cent}\bigl(r,\mathcal R_{r,r'}(\hat\pi_W)\bigr)
&\ge
\widetilde W^{\mathrm{block}}_{r'}(r,\hat\pi_W)\\
&\ge
\widetilde W^{\mathrm{block}}_{r'}(r,\bar\pi_W)\\
&=
\min_{v\in V}
g_{h,\lambda}\bigl(\cov_{r'}(v,\bar\pi_W)\bigr)\\
&\ge
\alpha_{r,r'}
\min_{v\in V}\cov_{\cent}(v,\Pi_W^*)\\
&=
\alpha_{r,r'}W_{\cent}^*(r).
\end{align*}
For constant \(r'\), the feasible region \(\mathcal Q_{r'}\) has polynomial size. The worst-case statement follows from the linear program \eqref{eq:block-worst-lp}. The average-case surrogate is an explicitly evaluable concave function over this polytope and can be maximized to any
prescribed additive accuracy. If its value is within \(\varepsilon\) of the surrogate optimum, the preceding chain loses at most \(\varepsilon\). Finally, \(F_{\cent}^*(r)\ge1\) for every nonempty instance, so the resulting value is at least \((\alpha_{r,r'}-\varepsilon)F_{\cent}^*(r)\).
\end{proof}
\begin{remark}
\label{rem:cent-block-hierarchy}
At the two endpoints, the construction recovers the familiar models. If \(r'=1\), then \(h=r\) and \(b=0\); the surrogate program is exactly the homogeneous problem, and \(\alpha_{r,1}=1-(1-1/r)^r=\alpha_r\). This factor is tight for the comparison with homogeneous repetition; see
Example~\ref{ex:cent-alpha-tight}. If \(r'=r\), then \(h=1\), \(b=0\), \(g_{1,0}(x)=x\), and the block program is the exact centralized problem, with \(\alpha_{r,r}=1\).
\end{remark}
The next lemma confirms that the guarantee improves monotonically with the block size, including at the values of \(r'\) where \(\lfloor r/r'\rfloor\) changes.
\begin{lemma}
\label{lem:alpha-increasing-rprime}
For fixed \(r\), the factor \(\alpha_{r,r'}\) is nondecreasing in \(r'\in[r]\).
\end{lemma}

\begin{proof}
Put \(t:=r/r'\ge1\), let \(h:=\lfloor t\rfloor\), and write \(t=h+\theta\) with \(0\le\theta<1\). Since \(b/r=\theta/t\), equation~\eqref{eq:block-factor} gives
\[
1-\alpha_{r,r'}
=q(t)
:=\frac{h}{t}\left(1-\frac1t\right)^h
=\frac{h(t-1)^h}{t^{h+1}}.
\]
On every interval \([h,h+1)\) with \(m\ge1\), and for \(t>1\),
\[
\frac{d}{dt}\log q(t)
=\frac{h}{t-1}-\frac{h+1}{t}
=\frac{h+1-t}{t(t-1)}
\ge0.
\]
Thus \(q\) is nondecreasing while \(h\) is fixed. There is no downward jump when \(h\) changes, because
\[
\lim_{t\uparrow h+1}q(t)
=\frac{h}{h+1}\left(\frac{h}{h+1}\right)^h
=\left(\frac{h}{h+1}\right)^{h+1}
=q(h+1).
\]
Hence \(q\) is nondecreasing on \([1,\infty)\). Increasing \(r'\) decreases \(t=r/r'\), so \(1-\alpha_{r,r'}=q(t)\) cannot increase. Therefore \(\alpha_{r,r'}\) is nondecreasing in \(r'\).
\end{proof}

\subsection{Hardness of Approximation}
\label{subsec:cent-hardness}

We now show that the \(1-\nicefrac{1}{e}\) approximation threshold persists in the centralized model, even under full-support simplex constraints. The two proofs use the same general principle: completeness is preserved by replacing deterministic marginals with sufficiently small full-support perturbations, while soundness follows by viewing every centralized strategy as a distribution over deterministic robot configurations.

\begin{proposition}
\label{prop:cent-avg-inapprox}
For every \(\delta>0\), the centralized average-case monitoring problem admits no polynomial-time
\((1-\nicefrac{1}{e}+\delta)\)-approximation unless
\(\mathrm{P}=\mathrm{NP}\), even when
\(\mathcal P=\Delta_V^\eta\) for a sufficiently small rational
\(\eta>0\).
\end{proposition}
\begin{proof}
It suffices to consider \(0<\delta<1/e\). Let
\(
\rho:=1-\frac1e+\delta
\)
and
\(
\beta:=1-\frac1e+\frac{\delta}{2}.
\)
We reduce from the gap version of \textsc{Maximum Domination}~\cite{MIYANO}. It is NP-hard to distinguish instances \((G,r)\) for which the maximum number of vertices dominated by at most \(r\) vertices is at least \(K\) from instances for which it is at most \(\beta K\).
Given such an instance, let \(n:=|V|\) and use the same graph \(G\) and the same number \(r\) of robots. Choose a rational number \(\eta\) satisfying
\(
0<\eta<\frac1n
\)
and
\(
rn^2\eta<\frac{\rho-\beta}{2\rho},
\)
and set \(\mathcal P:=\Delta_V^\eta\). Let \(\OPT_0\) denote the optimum value of \textsc{Maximum Domination} on \((G,r)\), and let \(\OPT_\eta:=F_{\cent}^*(r)\) denote the centralized average-case optimum under \(\Delta_V^\eta\).

\noindent We first compare these two values. For every centralized strategy \(\pi\), each configuration \(s=(s_1,\ldots,s_r)\) in the support of \(\pi\) monitors at most \(\OPT_0\) vertices. Therefore,
\[
F_{\cent}(r,\pi)
=
\mathbb E_{s\sim\pi}
\left[
  \left|\bigcup_{i=1}^r\N[s_i]\right|
\right]
\le \OPT_0,
\]
and hence \(\OPT_\eta\le\OPT_0\).
Conversely, let \(s_1,\ldots,s_r\) be an optimal deterministic solution to \textsc{Maximum Domination}, with repetitions if necessary, and define
\(
\mu_i:=\eta\mathbf 1+(1-n\eta)e_{s_i},
\)
for
\(
i\in[r].
\)
The product distribution with marginals \(\mu_1,\ldots,\mu_r\) is a feasible centralized strategy. Relative to the deterministic placement, the single-robot coverage probability of any vertex changes by at most \(n\eta\). Consequently, its monitoring probability under the \(r\)-robot product strategy changes by at most \(rn\eta\). Summing over all vertices gives
\(
\OPT_\eta
\ge
\OPT_0-rn^2\eta.
\)
Thus,
\(
\OPT_0-rn^2\eta
\le
\OPT_\eta
\le
\OPT_0.
\)
\noindent Suppose now that a polynomial-time \(\rho\)-approximation algorithm existed. In the YES case, it would return a solution of value at least
\[
\rho\OPT_\eta
\ge
\rho(\OPT_0-rn^2\eta)
\ge
\rho(K-rn^2\eta)
>
\beta K,
\]
where the strict inequality follows from the choice of \(\eta\) and the fact that \(K\ge1\). In the NO case, every centralized strategy has value at most
\(\OPT_0\le\beta K\). The algorithm would therefore distinguish the two cases of the gap \textsc{Maximum Domination} instance, contradicting \(\mathrm{P}\ne\mathrm{NP}\).
\end{proof}
We next show that allowing arbitrary correlations among robot locations does not improve the inapproximability threshold for the worst-case objective.
\begin{thm}
\label{thm:cent-worst-inapprox}
For every \(\delta>0\), the centralized worst-case monitoring problem admits no polynomial-time
\((1-\nicefrac{1}{e}+\delta)\)-approximation unless
\(\mathrm{P}=\mathrm{NP}\), even when
\(\mathcal P=\Delta_V^\eta\) for a sufficiently small rational
\(\eta>0\).
\end{thm}
\begin{proof}
The proof follows the same reduction as Theorem~\ref{thm:inapp-het-worst}; the pointwise soundness bound used there remains valid for arbitrary joint distributions. The complete proof is given in Appendix~\ref{appendix:section:centralized}.
\end{proof}
\subsection{Star-concavity} We next study how the centralized optimum evolves with the number of robots. As in the homogeneous case, both centralized objectives satisfy a decreasing-ratio property, or equivalently, $F_{\cent}^*(r)$ and $W_{\cent}^*(r)$ are star-concave in $r$. 

\begin{proposition}
\label{prop:cent-star-concavity}
Both $F_{\cent}^*(r)$ and $W_{\cent}^*(r)$ are star-concave in $r$, i.e.,
$
\frac{F_{\cent}^*(r+1)}{r+1}\le \frac{F_{\cent}^*(r)}{r}$ and 
$\frac{W_{\cent}^*(r+1)}{r+1}\le \frac{W_{\cent}^*(r)}{r}$.
\end{proposition}
\begin{proof}
Let \(\pi^{r+1}\in\Delta_{V^{r+1}}\), and assume that the marginal distribution \(\mu_i\) of each robot \(i\) belongs to \(\mathcal P\). Using \(\pi^{r+1}\), we construct a distribution \(\pi'\) over configurations of \(r\) positions. For \(s=(s_1,\ldots,s_{r+1})\in V^{r+1}\) and \(d\in[r+1]\), let
\(
\operatorname{Del}_d(s)
:=
(s_1,\ldots,s_{d-1},s_{d+1},\ldots,s_{r+1})
\in V^r
\)
denote the configuration obtained by deleting the \(d\)-th coordinate of \(s\). Define
\(
\pi'_{s'}
:=
\frac{1}{r+1}
\sum_{d=1}^{r+1}
\sum_{\substack{s\in V^{r+1}\\
\operatorname{Del}_d(s)=s'}}
\pi^{r+1}_s,
\)
for
\(
s'\in V^r.
\)
Equivalently, one samples \(s\sim\pi^{r+1}\), chooses an index \(d\in[r+1]\) uniformly at random, and deletes the \(d\)-th robot. Notice that different deletion indices are counted separately, even when they produce the same configuration \(s'\). The resulting masses sum to one, since
\[
\sum_{s'\in V^r}\pi'_{s'}
=
\frac{1}{r+1}
\sum_{d=1}^{r+1}
\sum_{s'\in V^r}
\sum_{\substack{s\in V^{r+1}\\
\operatorname{Del}_d(s)=s'}}
\pi^{r+1}_s
=
\frac{1}{r+1}
\sum_{d=1}^{r+1}
\sum_{s\in V^{r+1}}\pi^{r+1}_s
=1.
\]
For any vertex \(v\in V\), we can therefore write
\begin{align}
\cov_{\cent}(v,\pi')
&=
\frac{1}{r+1}
\sum_{s\in V^{r+1}}
\pi^{r+1}_s
\left|
\left\{
d\in[r+1]:
\operatorname{Del}_d(s)
\text{ covers }v
\right\}
\right|.
\label{eq:centsta}
\end{align}
If \(s\) does not cover \(v\), none of the resulting configurations covers \(v\). Suppose instead that \(s\) covers \(v\). If exactly one robot covers \(v\), then deleting any of the other \(r\) robots preserves coverage. If at least two robots cover \(v\), deleting any single robot preserves coverage. Consequently, every configuration \(s\) that covers \(v\) has at least \(r\) deletion indices for which \(\operatorname{Del}_d(s)\) still covers \(v\). It follows from
\eqref{eq:centsta} that
\begin{align}
\cov_{\cent}(v,\pi')
&\ge
\frac{r}{r+1}
\sum_{\substack{s\in V^{r+1}\\
\exists\,j\in[r+1]:s_j\in\N[v]}}
\pi^{r+1}_s
\nonumber\\
&=
\frac{r}{r+1}
\cov_{\cent}(v,\pi^{r+1}).
\label{eq:covc}
\end{align}
It remains to verify that every marginal distribution under \(\pi'\) belongs to \(\mathcal P\). Fix \(i\in[r]\) and \(v\in V\). If the deleted index satisfies \(d>i\), then the \(i\)-th coordinate of \(\operatorname{Del}_d(s)\) is \(s_i\). If \(d\le i\), then its \(i\)-th coordinate is \(s_{i+1}\). Therefore,
\begin{align}
\mu'_i(v)
&=
\frac{1}{r+1}
\left(
\sum_{d=i+1}^{r+1}
\sum_{\substack{s\in V^{r+1}\\s_i=v}}
\pi^{r+1}_s
+
\sum_{d=1}^{i}
\sum_{\substack{s\in V^{r+1}\\s_{i+1}=v}}
\pi^{r+1}_s
\right)
\nonumber\\
&=
\frac{r+1-i}{r+1}\mu_i(v)
+
\frac{i}{r+1}\mu_{i+1}(v).
\end{align}
Hence
\(
\mu'_i
=
\frac{r+1-i}{r+1}\mu_i
+
\frac{i}{r+1}\mu_{i+1}.
\)
Since \(\mathcal P\) is convex and
\(\mu_i,\mu_{i+1}\in\mathcal P\), we obtain
\(\mu'_i\in\mathcal P\). Thus, \(\pi'\) is a feasible centralized
strategy for \(r\) robots.
Now let \(\pi^{r+1}\) be optimal for the average-case objective.
Summing \eqref{eq:covc} over all vertices gives
\(
F_{\cent}^*(r)
\ge
F_{\cent}(r,\pi')
\ge
\frac{r}{r+1}F_{\cent}^*(r+1).
\)
Similarly, if \(\pi^{r+1}\) is optimal for the worst-case objective,
taking the minimum over \(v\) in \eqref{eq:covc} gives
\[
W_{\cent}^*(r)
\ge
W_{\cent}(r,\pi')
\ge
\frac{r}{r+1}W_{\cent}^*(r+1).
\]
Equivalently,
\[
\frac{F_{\cent}^*(r+1)}{r+1}
\le
\frac{F_{\cent}^*(r)}{r}
\quad
\text{and}
\quad
\frac{W_{\cent}^*(r+1)}{r+1}
\le
\frac{W_{\cent}^*(r)}{r}.\qedhere
\]
\end{proof}

\FloatBarrier

\section{Numerical Experiments and Illustrations}
\begin{figure}[t]
    \centering
    \begin{minipage}[t]{0.48\linewidth}
    \vspace{0pt}
        \centering
        \includegraphics[width=\linewidth]{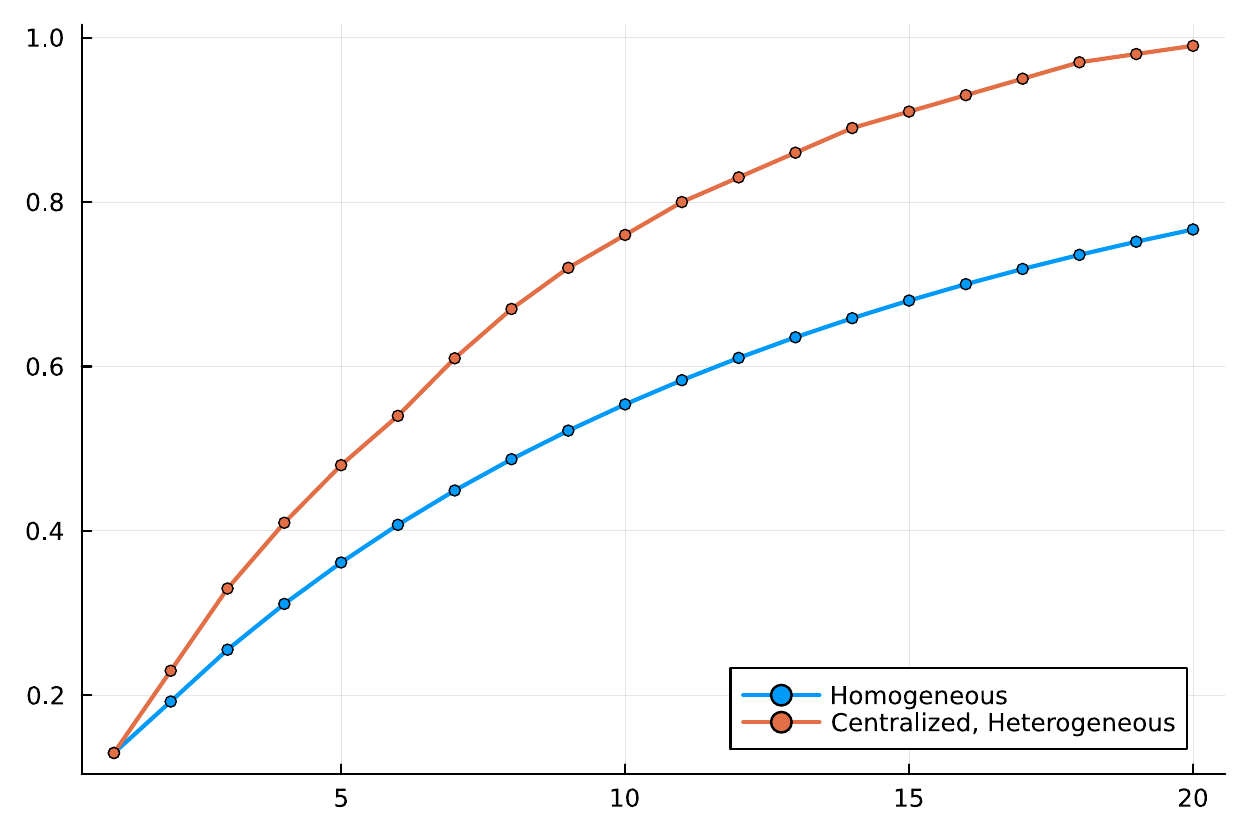}
        {\footnotesize (a) Average-case on \(G(100,0.05)\).}
        \label{fig:er-simplex-1}
    \end{minipage}
    \hfill
    \begin{minipage}[t]{0.48\linewidth}
    \vspace{0pt}
        \centering
        \includegraphics[width=\linewidth]{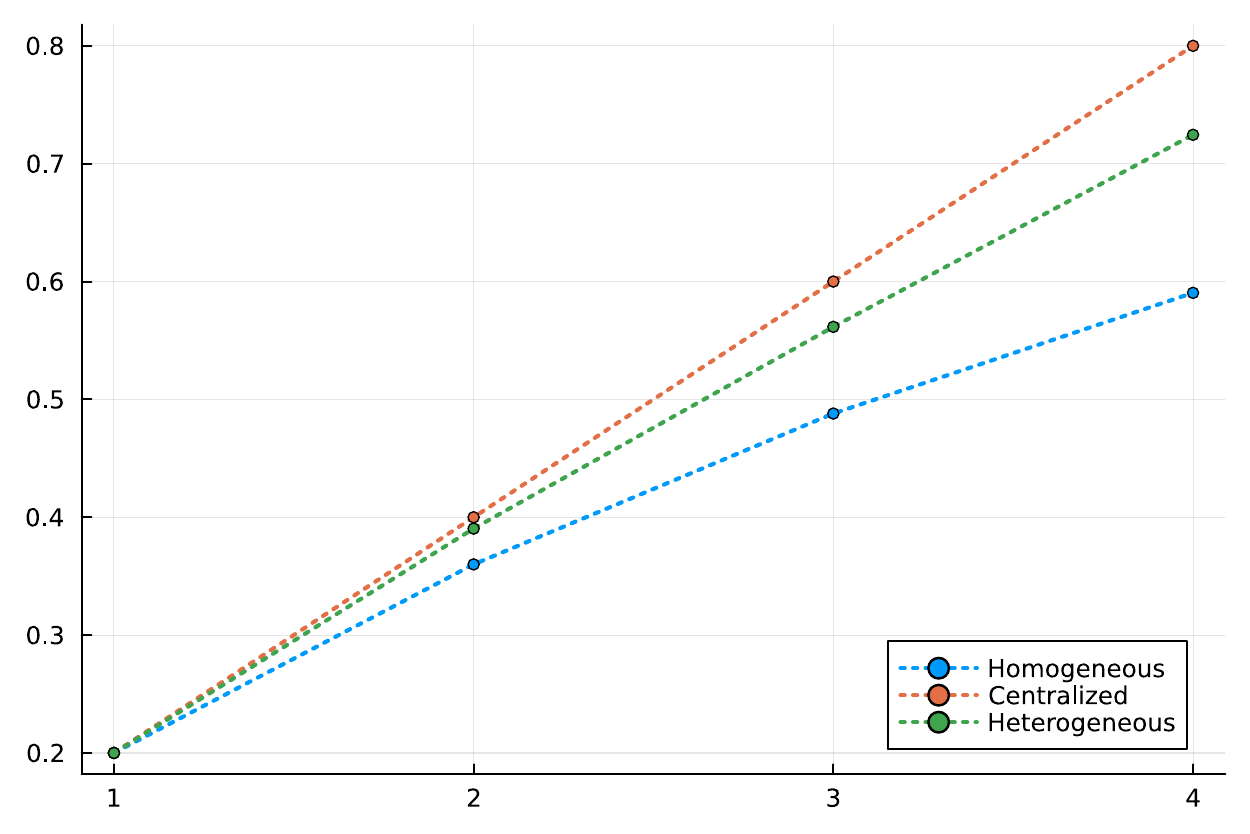}

        {\footnotesize (b) Worst-case on \(G(15,0.05)\).}
        \label{fig:er-simplex-2}
    \end{minipage}
    \vspace{0.5em}
    \begin{minipage}[t]{0.48\linewidth}
    \vspace{10pt}
        \centering
        \includegraphics[width=\linewidth]{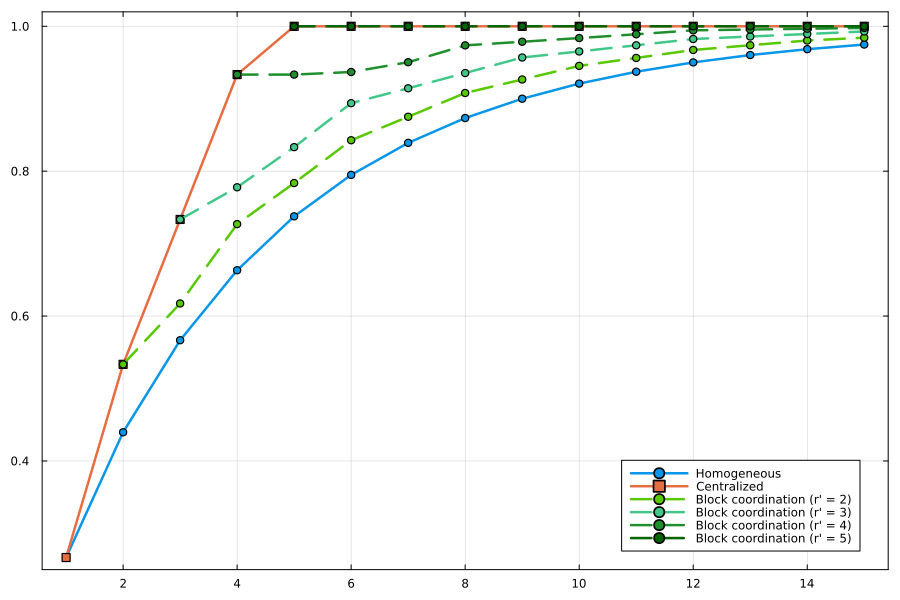}
        {\footnotesize (c) Average-case on \(G(15,0.05)\)}
        \label{fig:er-simplex-3}
    \end{minipage}
    \hfill
    \begin{minipage}[t]{0.48\linewidth}
    \vspace{10pt}
        \centering
        \includegraphics[width=\linewidth]{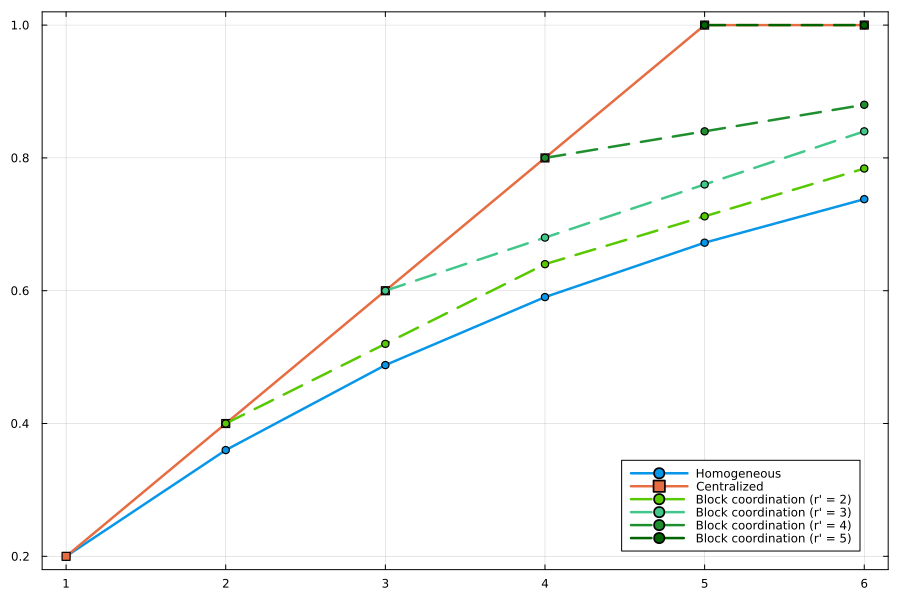}
        {\footnotesize (d) Worst-case on \(G(15,0.05)\).}
        \label{fig:er-simplex-4}
    \end{minipage}
    \caption{Monitoring performance on \text{Erd\H{o}s--R\'enyi} graphs \(G(n,p)\) with \(n\) vertices, where each edge is present independently with probability \(p\), conditioned on being connected, under simplex constraints, as a function of the number of robots \(r\). The x-axis shows \(r\), and the y-axis shows the optimal objective value.}
    \label{fig:er-simplex}
\end{figure}
We report numerical experiments for the unconstrained simplex baseline, on connected \text{Erd\H{o}s--R\'enyi} graphs \(G(n,p)\), interpreted as the limiting case of full-support simplex constraints \(\Delta_V^\eta\) as \(\eta\to 0\). Our goal is to compare the homogeneous, heterogeneous, centralized, and block-coordination models on representative graph instances, to quantify the gap between different models as the number of robots increases, and to assess the empirical performance of the block-coordination algorithm.

Throughout the experiments, average-case values are normalized by \(n\), so that all plotted quantities lie in \([0,1]\). Thus, the average-case curves represent the average monitoring probability per vertex, while the worst-case curves represent the monitoring probability of the least-covered vertex.

For the homogeneous model, we compute the average-case optimum by solving the convex formulation of Section~\ref{sec:hom} with MOSEK, and the worst-case optimum by solving the linear program of Section~\ref{sec:hom} with Gurobi. For the centralized model, the average-case optimum is computed through its equivalent Maximum Domination formulation under simplex constraints, whereas the worst-case optimum is obtained from the exact linear programming formulation over joint configurations.  For the heterogeneous model, the worst-case objective is evaluated by using the corresponding nonconvex quadratic formulations solved with Gurobi. Under simplex constraints, the heterogeneous and centralized average-case objectives match with Maximum Domination, and are therefore represented by a single curve shown in Figure~\ref{fig:er-simplex}(a), in contrast to the worst-case setting shown in Figure~\ref{fig:er-simplex}(b). Moreover, both figures illustrate the approximation guarantees, as the gap between the curves is consistent with the theoretical bounds. The average-case curves, shown in Figure~\ref{fig:er-simplex}(a), exhibit the diminishing-returns property as the number of robots increases.

For the block-coordinated algorithm, we solve the corresponding $r'$-robot problem for different block sizes and repeat the solution independently across groups. The resulting performance is shown in Figures~\ref{fig:er-simplex}(c) and \ref{fig:er-simplex}(d). Each value of $r'$ induces a separate curve, forming a continuum between the homogeneous case ($r'=1$) and the centralized solution ($r'=r$). As established theoretically, the curves interpolate smoothly between these two extremes.

\bibliographystyle{alpha}
\bibliography{lipics/ref.bib}

\appendix
\section{Appendix}
\subsection{Details of Section~\ref{sec:hom} (Independent Homogeneous Monitoring)}
\label{appendix:section:hom}

\begin{example}
\label{ex:f-hom-is-not-concave}
Consider \(P_4:v_3-v_2-v_1-v_4\). Fix a sufficiently small \(0<\eta<1/4\), and let \(\Delta_V^\eta\) be the line segment joining the two full-support distributions
\[
\mu^A_\eta=(\eta,1-3\eta,\eta,\eta)
\quad\text{and}\quad
\mu^B_\eta=(\eta,\eta,1/4,3/4-2\eta),
\]
written in the order \((v_1,v_2,v_3,v_4)\). Thus every \(\mu\in\Delta_V^\eta\) can be written as \(\mu^\eta_t=t\mu^A_\eta+(1-t)\mu^B_\eta\), for \(t\in[0,1]\), and every coordinate of \(\mu^\eta_t\) is strictly positive. Let \(F_{\mathrm{hom},\Delta_V^\eta}^*(r):=\max_{\mu\in\Delta_V^\eta}F_{\hom}(r,\mu)\).

To see that \(F_{\mathrm{hom},\Delta_V^\eta}^*(r)\) is not concave in \(r\), consider first the limiting values obtained as \(\eta\to0\). In the limit, the line segment is parameterized by \(\mu_t=(0,t,\frac{1-t}{4},\frac{3(1-t)}{4})\), and the closed-neighborhood masses are \(\cov(v_1,\mu_t)=\frac{3+t}{4}\), \(\cov(v_2,\mu_t)=\cov(v_3,\mu_t)=\frac{1+3t}{4}\), and \(\cov(v_4,\mu_t)=\frac{3(1-t)}{4}\). Hence \(F_{\hom}(r,\mu_t)=\sum_{v\in V}(1-(1-\cov(v,\mu_t))^r)\). For \(r=1\), this gives \(F_{\hom}(1,\mu_t)=2+t\), so the limiting optimum is \(3\). For \(r=2\), one obtains \(F_{\hom}(2,\mu_t)=-\frac74t^2+2t+\frac{11}{4}\), whose maximum over \([0,1]\) is attained at \(t=4/7\), giving the limiting value \(93/28\). For \(r=3\), one obtains \(F_{\hom}(3,\mu_t)=\frac{7}{16}t^3-3t^2+\frac{39}{16}t+\frac{25}{8}\), whose maximum over \([0,1]\) is attained at \(t=(16-\sqrt{165})/7\), giving the limiting value \((-687+165\sqrt{165})/392\). Numerically, these values are \(3\), \(93/28\approx3.3214\), and approximately \(3.6542\), respectively. Therefore the limiting values satisfy
\[
\frac{93}{28}
<
\frac12\left(3+\frac{-687+165\sqrt{165}}{392}\right).
\]
The objective is continuous in \((\eta,t)\), and the maximum over \(t\in[0,1]\) is therefore continuous as a function of \(\eta\). Since the above inequality is strict at \(\eta=0\), it remains true for all sufficiently small positive \(\eta\). Hence \(F_{\mathrm{hom},\Delta_V^\eta}^*(r)\) is not concave in \(r\), even though every distribution in \(\Delta_V^\eta\) has full support.
\end{example}

\begin{proof}[Proof of Proposition~\ref{thm:orbit-symmetric}]
For each \(\varphi\in\mathrm{Aut}(G)\) and \(u\in V\), define the relabeled distribution \(\mu^\varphi\) by \(\mu^\varphi(u):=\mu(\varphi^{-1}(u))\). Since \(\varphi\) preserves adjacency, it preserves closed neighborhoods, and therefore \(p_v(\mu^\varphi)=p_{\varphi^{-1}(v)}(\mu)\) for all \(v\in V\). It follows that both homogeneous objectives are invariant under relabeling: \(F_{\hhom}(\mu^\varphi)=F_{\hhom}(\mu)\) and \(W_{\hhom}(\mu^\varphi)=W_{\hhom}(\mu)\).

Now define the symmetrized distribution \(\bar\mu:=\frac{1}{|\mathrm{Aut}(G)|}\sum_{\varphi\in\mathrm{Aut}(G)}\mu^\varphi\). Clearly, \(\bar\mu\in\Delta_V\). Moreover, if \(\mu\) satisfies a uniform lower bound \(\mu(v)\ge\eta\) for all \(v\), then so does \(\bar\mu\). For the average-case objective, the map \(\mu\mapsto p_v(\mu)\) is linear and \(\phi_r\) is concave on \([0,1]\). Hence \(F_{\hhom}\) is concave in \(\mu\), and Jensen's inequality yields
\[
F_{\hhom}(\bar\mu)
\ge
\frac{1}{|\mathrm{Aut}(G)|}\sum_{\varphi\in\mathrm{Aut}(G)}F_{\hhom}(\mu^\varphi)
=
F_{\hhom}(\mu).
\]
For the worst-case objective, fix \(v\in V\). Again by linearity of \(p_v\) and concavity of \(\phi_r\),
\[
\phi_r(p_v(\bar\mu))
\ge
\frac{1}{|\mathrm{Aut}(G)|}\sum_{\varphi\in\mathrm{Aut}(G)}\phi_r(p_v(\mu^\varphi)).
\]
The right-hand side is an average of terms of the form \(\phi_r(p_w(\mu))\), and is therefore at least \(\min_{w\in V}\phi_r(p_w(\mu))=W_{\hhom}(\mu)\). Since this holds for every \(v\), taking the minimum over \(v\) gives \(W_{\hhom}(\bar\mu)\ge W_{\hhom}(\mu)\). Thus symmetrization does not decrease either objective.

It remains to show that \(\bar\mu\) is constant on each orbit. If \(u,u'\in O_i\), then there exists \(\psi\in\mathrm{Aut}(G)\) such that \(\psi(u)=u'\). Since \(\bar\mu\) is invariant under every automorphism, \(\bar\mu(u')=\bar\mu(\psi(u))=\bar\mu(u)\). Hence \(\bar\mu\) is constant on each orbit. Therefore, starting from any optimal solution and symmetrizing it, we obtain an optimal solution that is constant on every orbit.

Now fix orbits \(O_i,O_j\). If \(v,v'\in O_i\), then there exists \(\psi\in\mathrm{Aut}(G)\) with \(\psi(v)=v'\). Because automorphisms preserve adjacency and preserve each orbit set-wise, \(\psi\) induces a bijection between \(\N[v]\cap O_j\) and \(\N[v']\cap O_j\). Hence \(|\N[v]\cap O_j|=|\N[v']\cap O_j|\), so \(m_{ij}\) is well-defined. Finally, if \(\mu\) is orbit-symmetric with values \(x_1,\dots,x_k\), then for any \(v\in O_i\),
\[
p_v(\mu)
=
\sum_{j=1}^k \sum_{u\in \N[v]\cap O_j}\mu(u)
=
\sum_{j=1}^k |\N[v]\cap O_j|x_j
=
\sum_{j=1}^k m_{ij}x_j
=
p_i(x).
\]
Thus every vertex in the same orbit has the same closed-neighborhood mass, and the reduced formulas follow.
\end{proof}

\begin{proposition}[Complete bipartite graphs]
\label{prop:Kab}
Let \(G=K_{a,b}\) with \(a\neq b\). Under the simplex baseline, \(F_{\hhom}^*(r)=1+\max(a,b)\) if \(r=1\), and
\[
F_{\hhom}^*(r)
=
a+b-
\frac{(a-1)^r(b-1)^r}
{\left(a(b-1)^{\frac{r}{r-1}}+b(a-1)^{\frac{r}{r-1}}\right)^{r-1}}
\]
if \(r>1\). Under full-support simplex constraints \(\mu(v)\ge\eta\) for all \(v\), the corresponding optimum converges to the displayed value as \(\eta\to0\).
\end{proposition}

\begin{proof}
Let \((A,B)\) be the bipartition of \(K_{a,b}\), with \(|A|=a\) and \(|B|=b\). Since \(a\neq b\), there are two vertex orbits, namely \(A\) and \(B\). By Proposition~\ref{thm:orbit-symmetric}, it is enough to consider orbit-symmetric distributions. Let \(t\) be the total mass assigned to \(A\), so that \(1-t\) is the total mass assigned to \(B\). Then the closed-neighborhood masses are \(p_A=1-t+t/a\) and \(p_B=t+(1-t)/b\). Consequently, \(F_{\hom}(t)=a\,\phi_r(1-t+t/a)+b\,\phi_r(t+(1-t)/b)\).

If \(r=1\), then \(\phi_1(p)=p\), and \(F_{\hom}(t)=1+a+(b-a)t\). Hence the maximum over \(t\in[0,1]\) is attained at \(t=1\) if \(b>a\), and at \(t=0\) if \(a>b\). Thus the limiting value is \(F_{\hom}^*(K_{a,b})=1+\max(a,b)\). Under the full-support constraints \(\mu(v)\ge\eta\), the feasible interval for \(t\) is \([a\eta,1-b\eta]\), so the optimum converges to the same value as \(\eta\to0\).

Assume now that \(r>1\). Using \(\phi_r(p)=1-(1-p)^r\), we obtain
\[
\begin{aligned}
F_{\hom}(t)
&=
a+b-a\left(\frac{a-1}{a}t\right)^r-b\left(\frac{b-1}{b}(1-t)\right)^r\\
&=
a+b-\frac{(a-1)^r}{a^{r-1}}t^r-\frac{(b-1)^r}{b^{r-1}}(1-t)^r.
\end{aligned}
\]
Thus maximizing \(F_{\hom}(t)\) is equivalent to minimizing \(f(t):=\frac{(a-1)^r}{a^{r-1}}t^r+\frac{(b-1)^r}{b^{r-1}}(1-t)^r\), for \(0\le t\le1\). Since
\[
f''(t)
=
r(r-1)\left(
\frac{(a-1)^r}{a^{r-1}}t^{r-2}
+
\frac{(b-1)^r}{b^{r-1}}(1-t)^{r-2}
\right)>0,
\]
the function \(f\) is strictly convex on \([0,1]\). Hence it has a unique minimizer, characterized by \(f'(t)=0\), that is, \(\frac{(a-1)^r}{a^{r-1}}t^{r-1}=\frac{(b-1)^r}{b^{r-1}}(1-t)^{r-1}\). Equivalently,
\[
\frac{t}{1-t}
=
\left(
\frac{(b-1)^r/b^{r-1}}{(a-1)^r/a^{r-1}}
\right)^{1/(r-1)}
=
\frac{a(b-1)^{\frac{r}{r-1}}}{b(a-1)^{\frac{r}{r-1}}}.
\]
Therefore \(t^*=\frac{a(b-1)^{\frac{r}{r-1}}}{a(b-1)^{\frac{r}{r-1}}+b(a-1)^{\frac{r}{r-1}}}\). Substituting \(t=t^*\) into \(f(t)\) gives
\[
f(t^*)
=
\frac{(a-1)^r(b-1)^r}
{\left(a(b-1)^{\frac{r}{r-1}}+b(a-1)^{\frac{r}{r-1}}\right)^{r-1}}.
\]
This proves the displayed formula in the simplex baseline. Under the full-support constraints, the feasible interval for \(t\) is \([a\eta,1-b\eta]\), which converges to \([0,1]\) as \(\eta\to0\); therefore the full-support optimum converges to the same value.
\end{proof}

The limiting formula in Proposition~\ref{prop:Kab} also has the expected symmetric specialization when \(a=b\): in this case, the graph is vertex-transitive, the uniform distribution is optimal, and \(F_{\hom}^*(r)=2a\,\phi_r((a+1)/(2a))\).

\subsection{Details of Section~\ref{sec:het} (Independent Heterogeneous Monitoring)}
\label{appendix:section:het}

\begin{example}[Tightness of Theorem~\ref{thm:het-avg-alpha}]
\label{ex:tight-connected}
Fix \(r\ge2\), \(M\ge2\), and \(0<\eta<1/(rM)\). Construct a connected graph \(G_{r,M}=(V,E)\) as follows. For each \(j\in[r]\), let \(C_j\) be a clique of size \(M\), and choose a distinguished vertex \(c_j\in C_j\). Add the edges \(\{c_j,c_{j+1}\}\) for \(j=1,\dots,r-1\). Thus the cliques \(C_1,\dots,C_r\) are linked by a path through the distinguished vertices, \(V=\bigcup_{j=1}^r C_j\), and \(|V|=rM\).

We consider the full-support simplex constraint \(\Delta_V^\eta:=\{\mu\in\Delta_V:\mu(v)\ge\eta \text{ for all }v\in V\}\). For the heterogeneous problem, let \(\mu_j=\eta\mathbf 1+(1-|V|\eta)e_{c_j}\) for each \(j\in[r]\). Each \(\mu_j\) has full support. The corresponding limiting deterministic placement at \(c_1,\dots,c_r\) covers every vertex with probability \(1\), and replacing it by the full-support tuple changes each vertex monitoring probability by at most \(r|V|\eta\). Therefore \(F_{\het}^*(r)\ge |V|(1-r|V|\eta)=rM(1-r^2M\eta)\).

We now bound the homogeneous optimum. Let \(\mu\in \Delta_V^\eta\) be any feasible homogeneous distribution, and for each \(j\in[r]\), define \(x_j:=\mu(C_j)\). Then \(x_j\ge0\) and \(\sum_{j=1}^r x_j=1\). For every non-distinguished vertex \(v\in C_j\setminus\{c_j\}\), we have \(\N[v]=C_j\), and therefore its single-robot coverage mass is \(\cov(v,\mu)=x_j\). Hence its \(r\)-robot coverage probability is \(\phi_r(x_j)\), where \(\phi_r(x):=1-(1-x)^r\). Summing over all non-distinguished vertices gives
\[
\sum_{j=1}^r\sum_{v\in C_j\setminus\{c_j\}}\phi_r(x_j)
=
(M-1)\sum_{j=1}^r\phi_r(x_j).
\]
Since \(\phi_r\) is concave on \([0,1]\), Jensen's inequality yields \(\frac1r\sum_{j=1}^r\phi_r(x_j)\le \phi_r(\frac1r\sum_{j=1}^r x_j)=\phi_r(1/r)\). Therefore \((M-1)\sum_{j=1}^r\phi_r(x_j)\le(M-1)r\phi_r(1/r)\). The remaining \(r\) distinguished vertices contribute at most \(r\), since each coverage probability is at most \(1\). It follows that \(F_{\hom}^*(r)\le(M-1)r\phi_r(1/r)+r\).

Combining the previous bounds, we obtain
\[
\frac{F_{\hom}^*(r)}{F_{\het}^*(r)}
\le
\frac{(M-1)r\alpha_r+r}{rM(1-r^2M\eta)}
=
\frac{\alpha_r+\frac{1-\alpha_r}{M}}{1-r^2M\eta}.
\]
Consequently, for every fixed \(r\) and every \(\varepsilon>0\), choosing \(M\) sufficiently large and then choosing \(\eta>0\) sufficiently small gives a connected graph \(G_{r,M}\) such that \(F_{\hom}^*(r)/F_{\het}^*(r)\le\alpha_r+\varepsilon\). Thus the approximation factor \(\alpha_r=1-(1-\frac1r)^r\) is tight on connected graphs up to an arbitrarily small additive error under full-support simplex constraints.
\end{example}

\begin{figure}[t]
\centering
\begin{tikzpicture}[
    scale=1,
    every node/.style={font=\small},
    distinguished/.style={circle, fill=black, inner sep=1.8pt},
    ordinary/.style={circle, draw, fill=white, inner sep=1.4pt}
]

\node[distinguished,label=below:$c_1$] (c1) at (0,-0.72) {};
\node[ordinary] (a11) at (-0.85,-0.10) {};
\node[ordinary] (a12) at (-0.52,0.78) {};
\node[ordinary] (a13) at (0.52,0.78) {};
\node[ordinary] (a14) at (0.85,-0.10) {};
\node at (0,0.18) (d1) {$\cdots$};

\draw[gray]
(c1)--(a11) (c1)--(a12) (c1)--(a13) (c1)--(a14)
(a11)--(a12) (a11)--(a13) (a11)--(a14)
(a12)--(a13) (a12)--(a14)
(a13)--(a14)
(a12)--(d1) (a13)--(d1) (c1)--(d1);

\node at (0,1.25) {$C_1$};

\node[distinguished,label=below:$c_2$] (c2) at (4.2,-0.72) {};
\node[ordinary] (a21) at (3.35,-0.10) {};
\node[ordinary] (a22) at (3.68,0.78) {};
\node[ordinary] (a23) at (4.72,0.78) {};
\node[ordinary] (a24) at (5.05,-0.10) {};
\node at (4.2,0.18) (d2) {$\cdots$};

\draw[gray]
(c2)--(a21) (c2)--(a22) (c2)--(a23) (c2)--(a24)
(a21)--(a22) (a21)--(a23) (a21)--(a24)
(a22)--(a23) (a22)--(a24)
(a23)--(a24)
(a22)--(d2) (a23)--(d2) (c2)--(d2);

\node at (4.2,1.25) {$C_2$};

\node at (7.2,0) {$\cdots$};

\node[distinguished,label=below:$c_r$] (cr) at (10.4,-0.72) {};
\node[ordinary] (ar1) at (9.55,-0.10) {};
\node[ordinary] (ar2) at (9.88,0.78) {};
\node[ordinary] (ar3) at (10.92,0.78) {};
\node[ordinary] (ar4) at (11.25,-0.10) {};
\node at (10.4,0.18) (dr) {$\cdots$};

\draw[gray]
(cr)--(ar1) (cr)--(ar2) (cr)--(ar3) (cr)--(ar4)
(ar1)--(ar2) (ar1)--(ar3) (ar1)--(ar4)
(ar2)--(ar3) (ar2)--(ar4)
(ar3)--(ar4)
(ar2)--(dr) (ar3)--(dr) (cr)--(dr);

\node at (10.4,1.25) {$C_r$};

\draw[thick] (c1)--(c2);
\draw[thick] (c2)--(6.1,-0.72);
\draw[thick] (8.3,-0.72)--(cr);

\end{tikzpicture}
\caption{Graph \(G_{r,M}\): the cliques \(C_1,\dots,C_r\) are linked in a path through the vertices \(c_1,\dots,c_r\).}
\label{fig:tight-connected-example}
\end{figure}
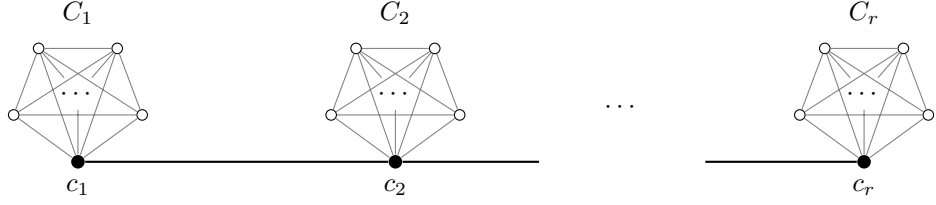

\begin{example}[Tightness of Corollary~\ref{cor:het-worst-alpha}]
\label{ex:tight-worst}
We consider the same graph \(G_{r,M}\) as in Example~\ref{ex:tight-connected}, under the full-support simplex constraint \(\Delta_V^\eta=\{\mu\in\Delta_V:\mu(v)\ge\eta \text{ for all }v\in V\}\). In the heterogeneous model, using the full-support distributions \(\mu_j=\eta\mathbf 1+(1-|V|\eta)e_{c_j}\) gives \(W_{\het}^*(r)\ge1-r|V|\eta\).

Let \(\mu\in \Delta_V^\eta\) be any homogeneous solution, and define \(x_j:=\mu(C_j)\). Since \(\sum_{j=1}^r x_j=1\), there exists \(j\) with \(x_j\le1/r\). For every vertex \(v\in C_j\setminus\{c_j\}\), we have \(\N[v]=C_j\), and hence \(v\) is covered with probability at most \(1-(1-1/r)^r=\alpha_r\). Therefore \(W_{\hom}^*(r)\le\alpha_r\). It follows that \(W_{\hom}^*(r)/W_{\het}^*(r)\le\alpha_r/(1-r|V|\eta)\). For every fixed \(r\) and every \(\varepsilon>0\), choosing \(\eta>0\) sufficiently small gives \(W_{\hom}^*(r)/W_{\het}^*(r)\le\alpha_r+\varepsilon\). Thus the factor is tight up to an arbitrarily small additive error under full-support simplex constraints.
\end{example}

\begin{example}
\label{exp:concavity-het-worst}
The decreasing-ratio property does not hold in general for the heterogeneous worst-case objective under full-support constraints. Consider the tree on the vertex set \(\{a,b,c,x,y,z\}\) with edge set \(\{\{x,a\},\{a,b\},\{b,c\},\{b,y\},\{c,z\}\}\). Fix \(0<\eta<1/6\), and let \(\Delta_V^\eta\) be the polytope of full-support distributions of the form \(\mu=\eta\mathbf 1+(1-6\eta)\bar\mu\), where \(\bar\mu\) is a distribution on \(\{a,b,c\}\). Equivalently, every feasible \(\mu\in \Delta_V^\eta\) assigns mass exactly \(\eta\) to each of \(x,y,z\), and the remaining mass is distributed over \(a,b,c\) with each of these vertices also receiving at least \(\eta\).

We prove the claim by comparing the full-support instance with its limit as \(\eta\to0\). In the limiting instance, for \(r=3\), placing the three robots at \(a,b,c\) covers every vertex with probability \(1\). Hence the limiting value is \(W_{\het,0}^*(3)=1\). For \(r=2\), let \(\bar\mu_1,\bar\mu_2\) be any two limiting feasible distributions on \(\{a,b,c\}\), and write \(p_u:=\bar\mu_1(u)\) and \(q_u:=\bar\mu_2(u)\) for \(u\in\{a,b,c\}\). The three leaves are covered only from their adjacent support vertices, so their coverage probabilities are \(h_u:=1-(1-p_u)(1-q_u)\), for \(u\in\{a,b,c\}\). Since \(\sum_u h_u=\sum_u(p_u+q_u-p_uq_u)=2-\sum_u p_uq_u\le2\), the minimum of the three values \(h_a,h_b,h_c\) cannot reach \(2/3\). Indeed, if all three were at least \(2/3\), then we would have \(h_a=h_b=h_c=2/3\) and \(\sum_u p_uq_u=0\). Thus \(p_uq_u=0\) for every \(u\), and \(h_u=2/3\) would imply \(p_u+q_u=2/3\) for every \(u\). Hence each \(p_u\) would belong to \(\{0,2/3\}\), contradicting \(p_a+p_b+p_c=1\). Therefore \(W_{\het,0}^*(2)<2/3\).

Consequently \(W_{\het,0}^*(2)/2<1/3=W_{\het,0}^*(3)/3\). The optimal values for \(\Delta_V^\eta\) are continuous in \(\eta\), because \(\Delta_V^\eta\) is an affine perturbation of the limiting feasible set and the objective is continuous. Since the inequality is strict in the limit, it persists for all sufficiently small positive \(\eta\). Hence the decreasing-ratio property fails for the heterogeneous worst-case objective under full-support constraints.
\end{example}

\subsection{Details of Section~\ref{sec:centralized} (Centralized Monitoring)}
\label{appendix:section:centralized}

\begin{proof}[Proof of Theorem~\ref{thm:cent-worst-inapprox}]
It suffices to consider \(0<\delta<1/e\). Let
\(
\rho:=1-\frac1e+\delta
\)
and
\(
\beta:=1-\frac1e+\frac{\delta}{2}.
\)
We use the same reduction from the gap version of
\textsc{Max \(k\)-Cover} as in
Theorem~\ref{thm:inapp-het-worst}. Recall that, in the constructed graph, every deterministic configuration of \(k\) robots monitors at most \(\beta|U|\) element-vertices in the NO case.

Set \(r:=k\), let \(n:=|V|\), and choose a rational number \(\eta\) such that
\(
0<\eta<\frac1n
\)
and
\(
\rho(1-rn\eta)>\beta.
\)
Finally, set \(\mathcal P:=\Delta_V^\eta\).

In the YES case, \(k\) sets cover all elements of \(U\). Place the robots on the corresponding set-vertices and replace each deterministic marginal \(e_{a_{j_i}}\) with
\(
\mu_i:=\eta\mathbf 1+(1-n\eta)e_{a_{j_i}}.
\)
The resulting product distribution is a feasible centralized strategy. The original deterministic configuration monitors every element-vertex and, because the set-vertices form a clique, every set-vertex. The full-support perturbation decreases the monitoring probability of any vertex by at most \(rn\eta\). Therefore,
\(
W_{\cent}^*(r)\ge1-rn\eta.
\)

Consider now the NO case, and let \(\pi\in\Delta_{V^r}\) be any feasible centralized strategy. For each \(e\in U\), define
\(
q_e
:=
\Pr_{s\sim\pi}
\bigl[b_e\text{ is monitored in }s\bigr].
\)
Since every deterministic configuration in the support of \(\pi\) monitors at most \(\beta|U|\) element-vertices,
\[
\sum_{e\in U}q_e
=
\mathbb E_{s\sim\pi}
\bigl[
  \#\{\text{element-vertices monitored in }s\}
\bigr]
\le
\beta|U|.
\]
Hence some \(e\in U\) satisfies \(q_e\le\beta\), and therefore
\(
W_{\cent}(r,\pi)\le q_e\le\beta.
\)
Since \(\pi\) was arbitrary, \(W_{\cent}^*(r)\le\beta\) in the NO case.

If a polynomial-time \(\rho\)-approximation algorithm existed, it would return a solution of value at least \(\rho(1-rn\eta)>\beta\) in the YES case, whereas every feasible solution would have value at most \(\beta\) in the NO case. This would distinguish the YES and NO cases of the gap \textsc{Max \(k\)-Cover} instance, contradicting \(\mathrm{P}\ne\mathrm{NP}\).
\end{proof}

\begin{example}
\label{ex:cent-alpha-tight}
For every integer \(r\ge1\), let \(n:=3r\), and let \(G=C_n\) be the cycle on \(V=\{1,2,\dots,n\}\), with indices taken modulo \(n\). Fix \(0<\eta<1/n\), and consider the full-support simplex constraint \(\Delta_V^\eta=\{\mu\in\Delta_V:\mu(v)\ge\eta \text{ for all }v\in V\}\). Define \(M:=\{2,5,8,\dots,3r-1\}\), that is, every third vertex around the cycle.

For the centralized solution, replace the deterministic placement on \(M\) by full-support marginals \(\mu_m=\eta\mathbf 1+(1-n\eta)e_m\), for \(m\in M\). The limiting deterministic placement has closed neighborhoods \(\N[2],\N[5],\dots,\N[3r-1]\), which form a partition of \(V\) into \(r\) disjoint blocks of size \(3\), and hence covers every vertex with probability \(1\). The full-support perturbation changes each vertex monitoring probability by at most \(rn\eta\). Therefore \(F_{\cent}^*(r)\ge n(1-rn\eta)\) and \(W_{\cent}^*(r)\ge1-rn\eta\).

\begin{figure}[h]
\centering
\begin{tikzpicture}[
    scale=1,
    every node/.style={font=\small},
    ordinary/.style={circle, draw, fill=white, inner sep=1.4pt},
    marked/.style={circle, fill=black, inner sep=1.8pt}
]

\node[ordinary,label={}] (v1) at (0,0) {};
\node[marked,label={}]   (v2) at (1,0) {};
\node[ordinary,label={}] (v3) at (2,0) {};
\draw[thick] (v1)--(v2)--(v3);
\draw[thick, rounded corners=6pt] (-0.5,-0.45) rectangle (2.5,0.45);
\node at (1,0.8) {$\{1,2,3\}$};

\node[ordinary,label={}] (v4) at (4,0) {};
\node[marked,label={}]   (v5) at (5,0) {};
\node[ordinary,label={}] (v6) at (6,0) {};
\draw[thick] (v4)--(v5)--(v6);
\draw[thick, rounded corners=6pt] (3.5,-0.45) rectangle (6.5,0.45);
\node at (5,0.8) {$\{4,5,6\}$};

\node at (8,0) {$\cdots$};

\node[ordinary,label={}] (vr1) at (10,0) {};
\node[marked,label={}]   (vr2) at (11,0) {};
\node[ordinary,label={}] (vr3) at (12,0) {};
\draw[thick] (vr1)--(vr2)--(vr3);
\draw[thick, rounded corners=6pt] (9.5,-0.45) rectangle (12.5,0.45);
\node at (11,0.8) {$\{3r-2,3r-1,3r\}$};

\draw[thick] (v3)--(v4);
\draw[thick] (v6)--(7.3,0);
\draw[thick] (8.7,0)--(vr1);
\draw[thick, bend left=20] (vr3) to (v1);

\end{tikzpicture}
\caption{Cycle \(C_{3r}\) grouped into blocks of three consecutive vertices. The black vertices are the elements of \(M\). Their closed neighborhoods are the blocks, which partition the cycle.}
\label{fig:cent-alpha-tight}
\end{figure}
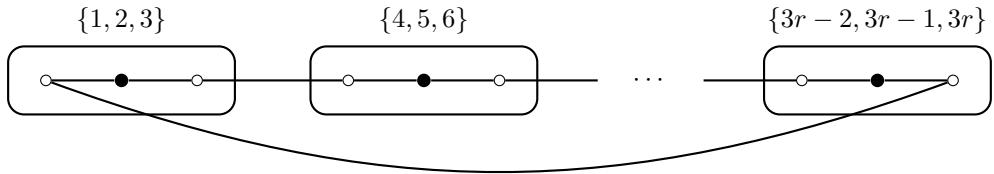

Since \(C_n\) is vertex-transitive, the uniform distribution is a maximizer for both homogeneous objectives. The uniform distribution is feasible for \(\Delta_V^\eta\) because \(\eta<1/n\), and symmetrization preserves the lower bounds. The closed neighborhood of each vertex has size \(3\), so \(F_{\hom}^*(r)=n\phi_r(3/n)=n\phi_r(1/r)=\alpha_r n\), and similarly \(W_{\hom}^*(r)=\phi_r(3/n)=\phi_r(1/r)=\alpha_r\). Combining these bounds gives
\[
\frac{F_{\hom}^*(r)}{F_{\cent}^*(r)}
\le
\frac{\alpha_r}{1-rn\eta},
\qquad
\frac{W_{\hom}^*(r)}{W_{\cent}^*(r)}
\le
\frac{\alpha_r}{1-rn\eta}.
\]
For every \(\varepsilon>0\), choosing \(\eta>0\) sufficiently small gives both ratios at most \(\alpha_r+\varepsilon\). Thus the \(\alpha_r\) comparison is tight up to an arbitrarily small additive error under full-support simplex constraints.
\end{example}

\end{document}